\documentclass[twoside,11pt]{article}
\usepackage{jmlr2e}

\usepackage{header}
\newtheorem{corollary}{Corollary}

\begin{document}
%
\ShortHeadings{Adaptive Incentive Design}{Ratliff and Fiez}
\title{Adaptive Incentive Design}
%
%
%

\author{Lillian J.~Ratliff  and Tanner Fiez\\ \addr{{\tt
    $\{$ratliffl,fiezt$\}$@uw.edu}\\ University of Washington\thanks{This work is supported by NSF Award CNS-1656873. T.~Fiez was also
    supported in part by an NDSEG Fellowship.}}}
\maketitle

\begin{abstract}
       We apply control theoretic and optimization techniques to adaptively design
   incentives. In particular,  we consider the problem of a \emph{\coordinator} with an objective that depends on data
   from strategic decision makers. The {\coordinator} 
does not know the process by which the strategic agents make decisions. 
Under the assumption that the agents are utility maximizers, we model their
interactions as a non--cooperative game and utilize the Nash
equilibrium concept as well as myopic update rules to model the selection of
their
decision. By parameterizing the agents' utility functions and the incentives
offered, we develop an algorithm that the {\coordinator} can employ to
learn the agents' decision-making processes while simultaneously designing
incentives to change their response to a more desirable response from the
{\coordinator}'s perspective. We provide convergence results for this
algorithm both in the noise-free and noisy cases and present illustrative examples.

\end{abstract}


%

\thispagestyle{empty}
\section{Introduction}
\label{sec:intro}
Due in large part to the increasing adoption of digital
technologies,
many applications that once treated users are passive entities must now consider
users as active participants.
In many application domains, a planner or coordinator, such as a platform
provider (e.g.,
transportation network companies), is tasked with
optimizing the performance of a system that people are actively interacting
with, often in real-time.  For instance, the planner may want to drive the system
performance to a more desirable behavior. While perhaps on competing ends of the
spectrum, both revenue maximization and social welfare maximization fall under
this umbrella. 

A significant challenge in optimizing such an objective is the
fact that human preferences are unknown \emph{a priori} and perhaps their
solicited responses, on which the system depends, may not be reported truthfully (i.e.~in
accordance with their true preferences) due to issues related privacy or trust.

We consider a class of incentive design problems in which a {\planner} 
does not know the underlying preferences, or decision--making process, of the agents that it is trying to coordinate. In the economics
literature these types of problems are known as problems of \emph{asymmetric
information}---meaning that the involved parties do not possess the same
information sets and, as is often the case, one party posses some information
 to which the other party is not privy.  
 
 The particular type of information
 asymmetry which we consider, i.e. where the preferences of the agents are
 unknown to the planner, results in a problem of \emph{adverse selection}. 
 The classic example of adverse selection is the \emph{market for
 lemons}~\cite{akerlof:1970aa} in which the seller of a used car knows more
 about the car than
 the buyer. There are a number of components that are hidden from the buyer such
 as the maintenance upkeep history, engine health, etc. Hence, the buyer could
 end up with a \emph{lemon} instead of a \emph{cherry}---i.e.~a broken down
 piece of junk versus a sweet ride.
Such problems have long been studied by economists.

The incentive design problem has also been explored by the control community,
usually in the context of (reverse) Stackelberg games (see,
e.g.,~\cite{Ho:1981aa,Ho:1984aa,Liu:1992aa}). 
More recently, dynamic incentive design in the context of
applications such as the power grid~\cite{Zhou:2017aa} or network congestion
games~\cite{Barrera:2015aa}.
We take a slightly different view by
employing techniques from learning and control to develop an adaptive method of
designing incentives in a setting where repeated decisions are made
by multiple, competiting agents whose preferences are unknown to the designer,
yet they are subjected to the incentives. 

 We assume that
agents, including the {\coordinator}, are cost minimizers\footnote{
    While in the remainder, we formulate the entire problem
given all agents are cost minimizers, the utility maximization formulation is
completely analogous.}. The decision space of the agents are assumed continuous.
We model each agent's cost as a parametric function that is dependent on the
choices of other agents and is modified by an incentive chosen by the planner.
 The {\coordinator} not knowing the underlying preferences of the
agents is tantamount to it not knowing  the value of the parameters 
of the agents' cost functions. Such parameters can be thought of as the
\emph{type} of the agent.

We formulate an adaptive incentive design problem in which the planner
iteratively learns the agents' preferences and optimizes the incentives offered
to the agents so as to drive them to a more desirable set of choices.
We derive an algorithm to solve this problem and provide theoretical results on convergence for both the case when the agents play according to a Nash equilibrium as well as the
case when the agents play myopically---e.g.~the agents play according to a
myopic update rule common in the theory of learning in
games~\cite{fudenberg:1998aa}. Specifically, we formulate an algorithm for iteratively estimating
preferences and designing incentives. By adopting
tools from adaptive control and online learning, we show that the
algorithm converges under reasonable assumptions.


The results have strong ties to both the adaptive control
literature~\cite{goodwin:1984aa,kumar:1986aa,sastry:1999aa} and the online
learning literature~\cite{cesa-bianchi:2006aa,nemirovski:2009aa,raginsky:2010aa}. The former gives us tools to do tracking of both the observed output
(agents' strategies) and the control input (incentive mechanism). It also allows
us to go one step further and prove parameter convergence under some
additional assumptions---\emph{persistence of excitation}---on the problem formulation and, in particular, the
utility learning and incentive design algorithm. The latter
provides tools that allow us to generalize the algorithm and get faster
convergence of the observed actions of the agents to a more desirable or even
socially optimal outcome. 

The remainder of the paper is organized as follows. We first introduce the
problem of interest in Section~\ref{sec:formulation}. In Sections~\ref{sec:incent_util}
and \ref{sec:incent_incent}, we mathematically formulate the utility learning and incentive
design problems and provide an algorithm for adaptive incentive design. We
present convergence results for the Nash and myopic-play cases in
Section~\ref{sec:main} after which we draw heavily on adaptive control techniques to
provide convergence results when the planner receives noisy observations. We
provide illustrative numerical examples in Section~\ref{sec:examples} and
conclude in Section~\ref{sec:discussion}.


\section{Problem Formulation}
\label{sec:problemformulation}
\label{sec:formulation}
We consider a problem in which there is a coordinator or planner with an objective, $f_P(x,v)$, that it
desires to optimize by selecting $v\in V\subset \mb{R}^n$; however, this objective is a function
of $x=(x_1, \ldots, x_n)\in X\subset \mb{R}^n$ which is the response of $n$ non--cooperative
strategic agents each providing response $x_i\in X_i\subset \mb{R}$. 
Regarding the dimension of each $x_i$, we assume without loss of generality that
they are scalars. 
All the theoretical results and insights apply to the more general setting where
each agent's choice is of arbitrary finite dimension.

The goal of the {\planner} is to design a mechanism to \emph{coordinate} the
agents by incentivizing them choose an $x$ that ultimately leads to minimization of
$f_P(x,v)$.
Yet, the
coordinator does not know the decision making process by which these agents
arrive at their collective response $x$. As a consequence there is \emph{asymmetric
information} between the agents and the {\planner}.

Let us suppose that the each agent has some type $\theta$ and a process
$\mc{M}_i^\theta$ that determines their \emph{choice}
$x_i(\theta_i)$. This process is dependent on the other agents and any mechanism
designed by the {\planner}. The classical approach in the economics literature
is to solve this problem of so-called \emph{adverse
selection}~\cite{bolton:2005aa} by designing mechanisms that induce agents' to
take actions in a way that corresponds with their true decision-making process
$\mc{M}_i^\theta$. In this approach, it is assumed that the coordinator has a
\emph{prior}
on the type space of the agents---e.g., a probability distribution on $\theta$.
The coordinator then designs a mechanism (usually static) based on this assumed
prior that encourages agents to act in accordance with their true preferences. 

We take an alternative view in which we adopt control theoretic and optimization
techniques to adaptively learn the agents' types while designing incentives to
coordinate the agents around a more desirable (from the point of view of the
{\planner}) choice $x$.
Such  a framework departs from  one-shot decisions that assume all prior
information is known at the start of the engagement and opens up opportunities
for mechanisms that are dynamic and can \emph{learn} over time.

We thus take the view that, in order to optimize its
objective, the {\planner} must learn the decision-making process and
simultaneously design a
mechanism that induces the agents to respond in such a way such  that the
{\planner}'s
objective is optimized. 

The {\coordinator} first optimizes its objective function to find the
\emph{desired} response $x^d$ and the \emph{desired} $v^d$. That is, it
determines the optimizers of its cost as if $x$ and $v$ are its decision
variables.
Of course, it may be the case that the set of optimizers of $f_P$
contains more than one pair $(x^d,v^d)$; in this case, the coordinator must
choose amongst the set of optimizers. In order to realize $(x^d,v^d)$, the
{\coordinator} must incentivize the agents to play $x^d$ by synthesizing  mappings
$\gamma_i\in \Gamma\subset C^2(X, \mb{R})$ for each $i\in \mc{I}=\{1, \ldots,
n\}$ such that $\gamma_i(x^d)=v_i^d\in V_i\subset \mb{R}$ and $x^d=(x_1^d,
\ldots, x_n^d)$ is the collective response of the agents under their true
processes $\{\mc{M}_i^\theta\}_{i\in \mc{I}}$.

We will consider two scenarios: (i) Agents play
according to a \emph{Nash equilibrium} strategy; (ii) Agents play according to a \emph{myopic}
update rule---e.g.~approximate gradient play or  fictitious play\cite{fudenberg:1998aa}.

In the first scenario, if the agents are assumed to play according to a Nash equilibrium strategy,
then $x^d$ must be a Nash equilibrium in the game induced by
$\gamma=(\gamma_1, \ldots, \gamma_\n)$. In particular, using the notation $x_{-i}=\{x_1,
\ldots, x_{i-1}, x_{i+1}, \ldots, x_n\}$, let agent $i$ have
\emph{nominal} cost $f_i\in C^2(X, \mb{R})$ and \emph{incentivized} cost
\eq{eq:incentcost}{
f_i^{\gamma_i}(x_i,
x_{-i})=f_i(x_i, x_{-i})+\gamma_i(x_i, x_{-i}).}
The desired response
$x^d$ is a Nash equilibrium of the incentivized game $(f_1^{\gamma_1}, \ldots,
f_n^{\gamma_n})$ if 
\begin{equation}
  x_i^d\in \arg\min_{x_i\in X_i} f_i^{\gamma_i}(x_i, x_{-i}^d). 
  \label{eq:bric}
\end{equation}
Hence, $x_i^d$ is a \emph{best response} to $x_{-i}^d$ for each
$i\in \mc{I}$. 
Formally, we define a Nash equilibrium as follows.

\begin{definition}[Nash Equilibrium of the Incentivized Game]
  A point $x\in X$ is a {Nash equilibrium} of the incentivized
  game $(\fg{1}, \ldots, \fg{n})$ if
  \begin{equation}
    \fg{i}(x)\leq \fg{i}(x_i', x_{-i}), \ \ \forall \
    x_i'\in X_i.
    \label{eq:nash}
  \end{equation}
\end{definition}
If, for each $i\in \mc{I}$, the inequality in \eqref{eq:nash} holds only for
a neighborhood $W_i\subset X_i$ of $x_i$, then $x$ is a \emph{local Nash
equilibrium}. 

We make use of a sub-class of Nash equilibria called
\emph{differential Nash equilibria}, as they can be characterized locally and
thus, amenable to computation. Let the \emph{differential game form}~\cite[Definition~2]{ratliff:2015aa}
$\omega:X\rar \mb{R}^n$ be defined by $\omega(x)=(D_1f_1(x), \ldots,
D_nf_n(x))$.
\begin{definition}[{\cite[Definition~4]{ratliff:2015aa}}]
  A strategy $x=(x_1, \ldots, x_n)\in X$ is a {differential Nash
  equilibrium} of $(\fg{1}, \ldots, \fg{n})$ if $\omega(x)=0$ and $D_{ii}^2f_i(x)$ is positive definite for
  each $i\in \pl$.
  \label{def:diffnash}
\end{definition}
Differential Nash equilibria are known to be generic amongst local Nash
equilibria~\cite{ratliff:2014aa}, structurally stable and attracting under
t\^{a}tonnement~\cite{ratliff:2015aa}.



In the second scenario, we assume the agents play according to a \emph{myopic
update rule}~\cite{fudenberg:1998aa} defined as follows.  
Given the incentive ${\gamma}_i$, agent $i$'s response
is determined by the mapping  
\begin{equation}
  \ggnt{i}(x)=g_i(x)+{\gamma}_i(x)
  \label{eq:gg}
\end{equation}
In addition, function $\ggnt{i}\in C^2(X^{k+1},\mb{R})$
maps the history, from time $0$ up to
time $k$, of the agents' previous collective response  to the
current response where $X^{k+1}$ is the
product space $X\times
\cdots \times X$ with $k+1$ copies of the space $X$.
We aim to design an algorithm in which the {\coordinator} performs a \emph{utility
learning} step and an \emph{incentive design} step such that as the
{\coordinator}
iterates through the algorithm, agents' collective observed response converges to the
desired response $x^d$ and the value of the incentive mapping evaluated at $x^d$ converges to 
the desired value $v^d$. In essence, we aim to ensure \emph{asymptotic} or
\emph{approximate} incentive compatibility. In the sections that follow, we describe the utility
learning and the incentive design steps of the algorithm and then, present the
algorithm itself.

%

\section{Utility Learning Formulation}
\label{sec:incent_util}
We first formulate a general utility learning problem, then we give examples in the the {\nplay} and
{\mplay} cases.


\subsection{Utility Learning Under Nash--Play}
\label{subsec:UL-nash}
We assume that the {\coordinator} knows the parametric structure of the agents'
nominal cost functions
and receives observations of the agents' choices over
time.
That is, for each $i\in\mc{I}$, we assume that the nominal cost function of agent $i$
has the form of a generalized linear model
\begin{equation}\textstyle
  f_i(x)=\braket{\Phi(x)}{\theta_i^\ast} = \sum_{j=1}^{m} \phi_{j}(x)\ts_{i,j}
  \label{eq:parametric}
\end{equation}
where $\Phi(x)$ is a vector of basis functions
 given by
$\Phi(x)=[\phi_{1}(x)\
\cdots \ \phi_{m}(x)]^T$, assumed to be known to the {\coordinator}, and
$\ts_i\in \mb{R}^m$
is a parameter vector, $\ts_i=[\ts_{i,1}\ \cdots\ \ts_{i, m}]^T$, assumed
unknown to the {\coordinator}. 

While our theory is developed for this case, we show through simulations
in Section~\ref{sec:examples}
that the {\planner} can be agnostic to the agents' decision-making processes and
still drive them to the desired outcome.

Let the set of basis functions for the agents' cost functions be denoted by $\mc{F}_\phi$. We assume that elements of $\mc{F}_{\phi}$ are $C^2(X, \mb{R})$ and \emph{Lipschitz continuous}. Thus the derivative of any function
  in $\mc{F}_\phi$ is uniformly bounded.

The admissible
set of parameters for agent $i$,  denoted by $\Theta_i$, is assumed to be a compact
subset of $\mb{R}^{m}$ and to contain the true parameter vector $\ts_i$. We will use the notation $f_i(x; \theta)$ when we need to 
make the dependence on the parameter
$\theta$ explicit. 

Note that we are limiting the problem of asymmetric information to
one of \emph{adverse selection}~\cite{bolton:2005aa} since it is the parameters
of the cost functions that are unknown to the coordinator. 

Similarly, we assume that the admissible incentive mappings have a generalized linear model of the form
\begin{equation}\textstyle
  \gamma_i(x) = \braket{\Psi(x)}{\g_i} = \sum_{j=1}^{\sn} \psi_{j}(x)\g_{i,j}
  \label{eq:gammaparam}
\end{equation}
where $\Psi(x)=[\psi_{1}(x)\ \cdots \ \psi_{\sn}(x)]^T$ is a vector of basis functions, belonging to a finite collection $\mc{F}_\psi$, and assumed to be $C^2(X, \mb{R})$ and \emph{Lipschitz continuous}, and
$\g_i=(\g_{i,1}\ \cdots\ \g_{i,\sn})\in \mb{R}^{\sn}$ are parameters.

\begin{remark}
This framework can be generalized to use different subsets of the basis
functions for different players, simply by constraining some of the parameters
$\theta_i$ or $\g_i$ to be zero. We choose to present the theory with a common
number of basis functions across players in an effort to minimize the amount of
notation that needs to be tracked by the reader.
\end{remark}

At each iteration $k$, the {\coordinator} receives the collective response from the
agents, i.e.~$x^k=(x_1^k, \ldots, x_n^k)$, and has the 
incentive parameters $\alpha^k = (\alpha_1^k, \dots, \alpha_n^k)$ that were issued. 


 We denote the set of
observations up to time $k+1$ by $\{\xx{t}\}_{t=0}^{k+1}$---where $x^0$ is the
observed Nash equilibrium of the nominal game (without incentives)---and the set of incentive
parameters~$\{\alpha^t\}_{t=0}^k$. 
Each of the observations is assumed to be an Nash equilibrium.

%

For the incentivized game
$(f_1^{\gamma_i}, \ldots, f_\n^{\gamma_\n})$, a Nash equilibrium $x$ necessarily
satisfies the first- and second-order conditions
$D_i\fg{i}(x)=0$ and
$D_{ii}^2\fg{i}(x)\geq 0$ for each $i\in\pl$ (see~\cite[Proposition~1]{ratliff:2015aa}).

Under this model, we assume that
the agents are playing a local Nash---that is, each $\xx{k}$ is a local Nash
equilibrium so that
\begin{equation}
  0=D_i\fg{i}(\xx{k})=\braket{D_i\Phi(\xx{k})}{\ts_i}+\braket{D_i\Psi(\xx{k})}{\ala{k-1}{i}},
  \label{eq:zero-cond-nash}
\end{equation}
for $k\in \mb{N}_+$ 
and $0=\braket{D_i\Phi(\xx{0})}{\ts_i}$,
where $D_i\Phi(\xx{k})=[D_{i}\phi_1(\xx{k})\ \cdots \
D_{i}\phi_m(\xx{k})]^T$ with $D_i\phi_j$ denoting the derivative of $\phi_j$ with
respect to $x_i$ and where  we define $D_i\Psi(\xx{k})$ similarly.
By an abuse of notation, we treat derivatives as vectors instead of co-vectors.

As noted earlier, without loss of
 generality, we take $x_i\in \mb{R}$. This makes the notation significantly
 simpler and the
 presentation of results much more clear and clean. All details for the general
 setting are provided in~\cite{ratliff:2015ab}. 

In addition, for each $i\in \pl$, we have
 \begin{align}
     0&\leq D_{ii}^2\fg{i}(\xx{k})=\braket{D_{ii}^2\Phi(\xx{k})}{
     \ts_i}+\braket{D_{ii}^2\Psi(\xx{k})}{\ala{k-1}{i}}
   \label{eq:secondord-cond}
 \end{align}
 for $k\in\mb{N}_+$ 
 and $0\leq \braket{D_{ii}^2\Phi(\xx{0})}{
 \ts_i}$ 
 where $D_{ii}^2\Phi$ and $D_{ii}^2\Psi$ are the second derivative of
 $\Phi$ and $\Psi$, respectively, with respect to $x_i$. 

Let the admissible set of $\theta_i$'s at iteration $k$ be denoted by
$\Theta_i^{k}$. They are defined using the second--order
conditions from the assumption that the observations at times $t\in\{1, \ldots,
k\}$ are local Nash equilibria and are
given by
\begin{align}
  {\Theta}^{k}_i=&\{\theta_i\in \Theta_i|\
      \braket{D_{ii}^2\Phi(\xx{t})}{\theta_i}+\braket{D_{ii}^2\Psi(\xx{t})}{
          \ala{t-1}{i}}\geq
  0,\notag\\
  &t\in\{1,\ldots, k\}\}\subseteq \Theta_i.
  \label{eq:thetak-2}
\end{align}
These sets are nested, i.e.~${\Theta}^{k}_i\subseteq  {\Theta}_i^{k-1}\subseteq \cdots
\subseteq  {\Theta}_i^{0}\subseteq \Theta_i,$
 since at each iteration an additional constraint is added
 to the previous set.
These sets are also convex since they are defined by semi{\ds}definite constraints~\cite{boyd:2004aa}.
Moreover, $\ts_i\in
\Theta_i^{k}$ for all $k$ since, by assumption, each observation $\xx{k}$ is
a local Nash equilibrium.

Since the {\coordinator} sets the incentives, 
given the response
$x^{k}$, they can compute the quantity $-\braket{D_i\Psi_{i}(\xx{k})}{\ala{k-1}{i}}$, which is equal to
\[
-\braket{D_i\Psi_{i}(\xx{k})}{\ala{k-1}{i}} = \braket{D_i\Phi_{i}(\xx{k})}{\ts_i}
\]
by the first order Nash condition~\eqref{eq:zero-cond-nash}. Thus, if we let 
$y_i^{k+1}=-\braket{D_i\Psi_i(\xx{k+1})}{\ala{k}{i}}$
and $\xi_i^{k}=D_i\Phi_i(\xx{k+1})$, we have
\begin{equation*}
    y_i^{k+1}=\braket{\xi_i^{k}}{\ts_i}.
  \label{eq:y}
\end{equation*}
Then, the coordinator has observations $\{y_i^t\}_{t=1}^{k+1}$ and \emph{regression
vectors} $\{\xi_i^t\}_{t=0}^{k}$.
We use the notation $\xi^k=(\xi_1^k, \ldots, \xi_n^k)$ for the regression
vectors of all the agents at iteration $k$.

\subsection{Utility Learning Under Myopic--Play}
\label{subsec:UL-myopic}
As in the Nash--play case, we assume the {\coordinator} knows the parametric
structure of the myopic update rule. That is to say, the nominal update function
$g_i$ is parameterized by $\theta_i$ over basis functions $\{{\phi}_{1},
\ldots, {\phi}_{m}\}\subset \mc{F}_\phi$ and the incentive mapping ${\gamma}_{i}^k$
at iteration $k$ is parameterized by
$\alpha_i^k$ over basis functions
$\{{\psi}_{1}, \ldots, {\psi}_{s}\}\subset \mc{F}_\psi$. 
We
assume the {\coordinator} observes the initial response $x^0$ and we denote
the past responses up to iteration $k$ by $x^{(0,k)}=(x^0,\ldots, x^k)$.
The general architecture for the myopic update rule is given by
\begin{align}
  \bx{i}{k+1}&=g_i^{{\gamma_i}}(\bx{i}{(0,k)},
 \bx{-i}{(0,k)})\\
 &=\braket{{\Phi}_i(\bx{i}{(0,k)},
 \bx{-i}{(0,k)})}{\theta_i^\ast}+\braket{{\Psi}_i(\bx{i}{(0,k)},
 \bx{-i}{(0,k)})}{\ala{k}{i}}.
  \label{eq:myopic_update}
\end{align}
Note that the update rule does not need to depend on the whole sequence of past
response. It could depend just on the past response $x^k$ or a subset, say
$x^{(j,l)}$ for $0\leq j\leq l\leq k$.

As before, we denote the set of admissible parameters for player
$i$ by $\Theta_i$ which we assume to be a compact subset of $\mb{R}^{m}$.
In contrast to the {\nplay} case, our admissible set  of
parameters is no long time varying so that $\Theta_i^k=\Theta_i$ for all $k$.

Keeping consistent with the notation of the previous sections, we let
$\xi_i^k={\Phi}_i(\xx{k})$  and 
$y_i^{k+1}=\bx{i}{k+1}-\braket{{\Psi}_i(\xx{k})}{\ala{k}{i}}$ 
so that the myopic update rule can be re-written as
$y_i^{k+1}=\braket{\xi_i^k}{\ts_i}.$
Analogous to  the previous case, 
the coordinator has observations $\{y_i^t\}_{t=1}^{k+1}$ and regression
vectors $\{\xi_i^t\}_{t=0}^{k}$.
Again, we use the notation $\xi^k=(\xi_1^k, \ldots, \xi_n^k)$ for the regression
vectors of all the agents at iteration $k$. 

Note that the form of the myopic update rule is general enough to accommodate a
number of game-theoretic learning algorithms including approximate fictitious play and
gradient play~\cite{fudenberg:1998aa}.

\subsection{Unified Framework for Utility Learning}
We can describe both the Nash--play and myopic--play cases in a unified
framework as follows. At iteration $k$, the {\coordinator} receives a response
$x^{k}$ which lives in the set $\mc{X}_k(x^{(0,k-1)},
\theta^\ast,\alpha^{k-1})$ which is either the set of local Nash equilibria of the
incentivized game  or the unique response determined by the
incentivized myopic update rule at iteration $k$. The {\coordinator} uses the past
responses $\{x^t\}_{t=0}^{k}$ and incentive parameters $\{\alpha^t\}_{t=0}^{k}$
to generate the set of observations $\{y^t\}_{t=1}^{k+1}$ and regression vectors
$\{\xi^t\}_{t=0}^k$. 

The utility learning problem is formulated as an online optimization problem in
which parameter updates are calculated following the gradient of a loss
function.
For each $i\in\pl$, consider
the loss function given by
\eq{eq:loss_function}{\textstyle
    \ell(\tit{i}{k})=\frac{1}{2}\|y^{k+1}_i-\braket{\xi_i^k}{\tit{i}{k}}\|_2^2} 
that evaluates the error between the predicted observation and the true observation
at time $k$
for each player.

In order to minimize this loss, we introduce a well-known generalization of the projection
operator.
Denote by $\partial f(x)$ the set of 
\emph{subgradients} of $f:X\rar\mb{R}$ at $x$. A convex continuous function $\beta:\Theta\rar\mb{R}$ is a \emph{distance
generating} function with modulus $\nu>0$ with respect to a reference norm $\|\cdot\|$, if the set 
$\Theta^\circ=\{\theta\in \Theta| \partial \beta(\theta)\neq
\emptyset\}$
is convex and restricted to $\Theta^\circ$, $\beta$ is continuously differentiable and
\emph{strongly convex} with
parameter $\nu$, that is
\begin{equation*}\braket{\theta'-\theta}{\nabla\beta(\theta')-\nabla\beta(\theta)}
\geq
\nu\|\theta'-\theta\|^2, \ \ \forall \ \theta', \theta\in \Theta^\circ.
\end{equation*}
The function $V:\Theta^\circ\times \Theta\rar \mb{R}_+$, defined by
\begin{equation*}
  V(\theta_1, \theta_2)=\beta(\theta_2)-\left(
  \beta(\theta_1)+\nabla\beta(\theta_1)^T(\theta_2-\theta_1) \right)
  \label{eq:Vdef-1}
\end{equation*}
is the Bregman divergence~\cite{bregman:1967aa}
associated
with $\beta$. By definition, $V(\theta_1,
\cdot)$ is non-negative and strongly convex with modulus $\nu$. Given a subset $\Theta^k \subset \Theta$ and a point $\theta \in \Theta^k$, the mapping
 $P_{\Theta^k, \theta}:\mb{R}^{m}\rar
\Theta^{k}$ defined by
\begin{equation}
P_{\Theta^k, \theta}(g)=\arg\min_{\theta'\in \Theta^{k}}\left\{
 \braket{g}{\theta'-\theta}+V(\theta, \theta') \right\}.
  \label{eq:proxmap-def-theta}
\end{equation}
is the prox-mapping induced by $V$ on $\Theta^k$. This mapping is well-defined,
the minimizer is unique by strong convexity of $V(\theta, \cdot)$, and is a
contraction at iteration $k$,~\cite[Proposition~5.b]{moreau:1965aa}.



Given the loss function $\ell: \Theta \to \mb{R}$, a positive, non-increasing sequence of learning rates $(\eta_k)$, and a distance generating function $\beta$, the parameter estimate of each $\theta_i$ is updated at iteration $k$ as follows 
\begin{equation}
  \tit{i}{k+1}=P_{\Theta_i^k, \tit{i}{k}}^{k+1}\left( \eta_k\nabla \ell(\tit{i}{k}) \right).
  \label{eq:theta_up_prox}
\end{equation}
Note that if the distance generating function is $\beta(\theta)=
\frac{1}{2}\|\theta\|_2^2$, then the associated Bregman divergence is the Euclidean distance, $V(\theta, \theta') = \frac{1}{2}\|\theta - \theta'\|_2^2$, and the corresponding prox--mapping is the Euclidean
projection on the set $\Theta_i$, which we denote by
$P_{\Theta_i, \theta}(g)=\Pi_{\Theta_i}(\theta - g)$, so that
\begin{equation}
  \tit{i}{k+1}=\Pi_{\Theta_i^{k}}\left( \tit{i}{k}-\eta_k\nabla
  \ell(\tit{i}{k}) \right).
  \label{eq:theta_up}
\end{equation}

\section{Incentive Design Formulation}
\label{sec:incent_incent}
\label{sec:incent-incent}
In the previous section, we described the 
 parameter update step that will be used in our utility learning and incentive
 design problem. We now describe how the incentive parameters $\alpha^k$ for
 each iteration are selected. 
In particular, at iteration $k$, after updating parameter estimates for each agent, the data
the
{\coordinator} has includes the past observations $\{y^t\}_{t=1}^{k+1}$, incentive
parameters $\{\alpha^t\}_{t=0}^k$, and has an estimate of each
$\tit{i}{k+1}$ for $i\in\pl$. 
The {\coordinator} then uses the past data along with the 
parameter estimates to find an $\al{k+1}=(\ala{k+1}{1},
\ldots, \ala{k+1}{\n})$ such that the incentive mapping for each player
evaluates to $v_i^d$ at $x_i^d$ and 
$x^d\in \mc{X}_{k+1}(x^{(0,k+1)}, \theta^\ast, \alpha^{k+1})$.
This is to say that if the agents are rational and play Nash, then 
$x^d$ is a local Nash equilibrium of the game 
$(\fg{1}(x; \tit{1}{k+1}), \ldots, \fg{\n}(x; \tit{\n}{k+1}))$
where $\fg{i}(x;\tit{i}{k+1})$ denotes the incentivized cost of player $i$
parameterized by $\tit{i}{k+1}$. On the other hand, if the agents are myopic,
then, for each
$i\in\pl$,
$x^d_i=\braket{{\Phi}_i(\xx{k})}{\tit{i}{k+1}}+\braket{{\Psi}_i(\xx{k})}{\ala{k+1}{i}}.$

In the following two subsections, for each of these cases, we describe how
$\al{k+1}$ is selected.

\subsection{Incentive Design: Nash--Play}
Given that
$\gamma_i^{k+1}$ is parameterized by $\alpha^{k+1}_i$, the goal is to
find $\ala{k+1}{i}$ for each $i\in\pl$ such that $x^d$ is a local Nash equilibrium of the game 
\begin{align*}
  &  \big(\braket{\Phi_1(x)}{\tit{1}{k+1}}+\braket{\Psi_1(x)}{\ala{k+1}{1}},
    \ldots,\braket{\Phi_\n(x)}{\tit{\n}{k+1}}+\braket{\Psi_\n(x)}{\ala{k+1}{\n})}\big)
\end{align*}
and such that $\braket{\Psi_i(x^d)}{\ala{k+1}{i}}=\y_i^d$ for each $i\in\pl$. 


\begin{assumption}
  For every $\{\theta_i\}_{i=1}^\n$ where $\theta_i\in \Theta_i$, there exist
  $\alpha_i\in \mb{R}^{s}$ for
    each $i\in\pl$ such that
  $x^d$ is the induced differential Nash equilibrium in the game 
  $(\fg{1}(x;\theta_1), \ldots, \fg{\n}(x;\theta_\n))$
  and $\gamma_i(x^d)=\y^d_i$ where $\gamma_i(x)=\braket{\Psi_i(x)}{\alpha_i}$. 
  \label{ass:nash-1}
\end{assumption}

We remark that the above assumption is not restrictive in the following sense.
  Finding $\alpha_i$ that induces the desired Nash equilibrium and results in
  $\gamma_i$ evaluating to the desired incentive value amounts to finding
  $\ala{k+1}{i}$ such that the first{\ds} and second{\ds}order sufficient
  conditions for a local Nash equilibrium are satisfied given our estimate of
  the agents' cost functions. 
  That is, for each $i\in\pl$, we need to find $\ala{k+1}{i}$ satisfying 
  \begin{equation}
    \left\{
    \begin{array}{ll}
      0_{2\times 1}
      &=\zeta_i^{k+1}+\Lambda_i\ala{k+1}{i}\\
0&< \braket{D_{ii}^2\Phi_i(x^d)}{
    \tit{i}{k}}+\braket{D_{ii}^2\Psi_i(x^d)}{\ala{k+1}{i}}
 \end{array}
 \right.
    \label{eq:controlfind-n}
  \end{equation}
  where
  \begin{equation*}
      \zeta_i^{k+1}=\bmat{\braket{D_i\Phi_i(x^d)}{\tit{i}{k+1}}\\ -\y_i^d}
  \end{equation*}
  and
  \begin{equation*}
      \Lambda_i=\bmat{D_i\Psi_i(x^d)^T\\
      \Psi_i(x^d)^T}\in \mb{R}^{2\times s}
  \end{equation*}
 If $\Lambda_i$ is full rank, i.e. has rank $2$, then there exists a
$\ala{k+1}{i}$ that solves the first equation in \eqref{eq:controlfind-n}. If the number of
basis functions $\sn$ satisfies $\sn>2$, then the rank condition is not
unreasonable and in fact, there are multiple solutions. In essence, by selecting
$\sn$ to be \emph{large enough}, the {\coordinator} is allowing for enough degrees of freedom
to ensure there exists a set of parameters $\alpha$ that induce the desired
result. Moreover, the problem of
finding $\ala{k+1}{i}$ reduces to a convex feasibility problem.

  The convex feasibility problem defined by~\eqref{eq:controlfind-n}
  can be formulated as a constrained least{\ds}squares optimization
  problem. Indeed, for each $i\in\mc{I}$,
  \begin{equation}
      \text{(P1)} \left\{   \begin{array}{ll}
      \min\limits_{\ala{k+1}{i}} & 
      \|\zeta_i^{k+1}+\Lambda_i\ala{k+1}{i}\|_2^2\notag\\
      \text{s.t.} &\braket{D_{ii}^2\Phi_i(x^d)}{
          \tit{i}{k+1}}+\braket{D_{ii}^2\Psi_i(x^d)}{\ala{k+1}{i})}\geq \vep
      \end{array}\right.
    \label{eq:findalpha-lsq}
  \end{equation}
for some $\vep>0$.  By Assumption~\ref{ass:nash-1}, for each $i\in\pl$, there is an
  $\alpha^{k+1}$ such that the cost is exactly minimized.
  
  The choice of $\vep$ determines how well-conditioned the
second-order derivatives of agents' costs with respect to their own choice
variables is. In addition, we note that if there are a large number of incentive
basis functions, it may be reasonable to incorporate a cost for sparsity---e.g.,
$\lambda_i\|\alpha_i^{k+1}\|_1$; however, the optimal solution in this case is
not guaranteed to satisfy \eqref{eq:controlfind-n}.

  
 It is desirable for the induced local Nash equilibrium to be a
  \emph{stable, non-degenerate
  differential Nash equilibrium} so that it is attracting in a neighborhood
  under the gradient flow~\cite{ratliff:2015aa}. 
  To enforce this, the {\coordinator} must add
  additional constraints to the feasibility problem defined by
  \eqref{eq:controlfind-n}. In particular, 
 second{\ds}order
  conditions on player cost functions must be satisfied, i.e.~that the
  derivative of
  the differential game form $\omega$ is positive--definite~\cite[Theorem~2]{ratliff:2015aa}. 
This reduces to
  ensuring
    $D^2\Phi(x, \tu{k+1})+D^2\Psi(x, \al{k+1})> 0$
where
\begin{align*}
 & D^2\Phi(x, \tu{k+1})=\bmat{\braket{D_{11}^2\Phi_1(x)}{\tit{1}{k+1}}
   &\cdots &
   \braket{D_{\n1}^2\Phi_1(x)}{\tit{1}{k+1}}\\
  \vdots&\ddots & \vdots \\
  \braket{D_{1\n}^2\Phi_\n(x)}{\tit{\n}{k+1}}
  & \cdots &
  \braket{D_{\n\n}^2\Phi_\n(x)}{\tit{\n}{k+1}}},
    \label{eq:dsqphi}
  \end{align*}
  and $D^2\Psi(x, \alpha^{(k+1)})$ is defined analogously.
 Notice that this constraint is a semi{\ds}definite 
 constraint~\cite{boyd:2004aa} and thus, the problem of finding $\al{k+1}$ that induces $x^d$ to be a
 stable, non{\ds}degenerate differential Nash equilibrium can be formulated as a
 constrained least{\ds}squares optimization
  problem. Indeed,
  \begin{equation*}
      \text{(P2)}      \left\{\begin{array}{ll}
          \min\limits_{\alpha^{k+1}} & 
          \sum_{i}\|\zeta_i^{k+1}+\Lambda_i\alpha_i^{k+1}\|_2^2\\
\text{s.t.} &D^2\Phi(x^d, \tu{k+1})+D^2\Psi(x^d, \al{k+1})\geq \vep\end{array}\right.
    \label{eq:fa-lsq-s}
  \end{equation*}
  for some $\vep>0$.
     The
  optimization problem~(P2) can be written as a semi{\ds}definite
  program. 
  
  Again, a regularization term can be incorporated in order to find
  sparse parameters $\ala{k+1}{i}$. However, by introducing regularization, the
  condition $\zeta_i^{k+1}+\Lambda_i\alpha_i^{k+1}=0$ will in general no longer
  be satisfied by the solution.

  Ensuring the desired Nash equilibrium is a stable, non{\ds}degenerate
  differential Nash equilibrium means that, first and foremost, the desired Nash
  equilibrium is \emph{isolated}~\cite[Theorem~2]{ratliff:2015aa}. Thus, there is no nearby Nash equilibria. Furthermore, non{\ds}degenerate differential Nash
  equilibria are generic~\cite[Theorem~1]{ratliff:2014aa} and structurally
  stable~\cite[Theorem~3]{ratliff:2015aa} so that they are robust to
  small modeling errors and environmental noise. Stability ensures that
  if at each iteration players play according to a myopic approximate best response strategy
  (gradient play), then they will converge to the desired Nash
  equilibrium~\cite[Proposition~2]{ratliff:2015aa}.
  Hence, if a stable equilibrium is desired by the {\coordinator}, we can
  consider a modified version of Assumption~\ref{ass:nash-1}.
\begin{customass}{1'}[Modified---Stable Differential Nash]
   For every $\{\theta_i\}_{i=1}^\n$ where $\theta_i\in \Theta_i$, there exist
   $\alpha_i\in \mb{R}^{\sn}$ for
    each $i\in\pl$ such that
  $x^d$ is the induced stable, non--degenerate differential Nash equilibrium of 
  $(\fg{1}(x;\theta_1), \ldots, \fg{\n}(x;  \theta_\n))$
  and $\gamma_i(x^d)=\y^d_i$ where $\gamma_i(x)=\Psi_i(x)^T\alpha_i$. 
\label{ass:nash-1-mod}
\end{customass}
Given $\tu{k}=(\tit{1}{k}, \ldots, \tit{\n}{k})$, 
let $\mc{A}(\tu{k}, x, \y)$
  be the set of
$\al{k}=(\ala{k}{1}, \ldots,\ala{k}{\n})$ such that $x$ is a 
 differential Nash equilibrium of  $(\fg{1}(x), \ldots, \fg{\n}(x))$
 and $\gamma_i(x)=\y_i$ where $\gamma_i(x)=\braket{\Psi_i(x)}{\ala{k}{i}}$.
Similarly, let $\mc{A}_s(\tu{k}, x, \y)$ be the set of
$\al{k}$ that induce $x$ to be a stable, non{\ds}degenerate differential
Nash equilibrium where $\gamma_i(x)=\y_i$. By Assumptions~\ref{ass:nash-1}
and~\ref{ass:nash-1-mod}, 
$\mc{A}(\tu{k}, x, \y)$
and $\mc{A}_s(\tu{k}, x, \y)$, respectively, are non{\ds}empty. Further,
it is straightforward to find an $\al{k}$ belonging to
$\mc{A}(\tu{k}, x, \y)$ (resp., $\mc{A}_s(\tu{k}, x,
\y)$) by solving the convex problem stated in~(P1)
(resp.,~(P2)). 
\subsection{Incentive Design: Myopic--Play}
Given $\theta_{i}^{k+1}$ for each
$i\in \pl$, the {\coordinator} seeks an incentive mapping
${\gamma}^{k+1}=(\ga{1}{k+1}, \ldots,
\ga{\n}{k+1})$ that induces the desired response $x^d$ and such
that 
$\ga{i}{k+1}(x^d)=\y_i^d$ for each $i\in\pl$. As before, given that
${\gamma}_i$ has been parameterized, this
amounts to finding $\ala{k+1}{i}$ such that 
\eq{eq:update-1}{
    x_i^d=\braket{{\Phi}(\xx{k})}{\tit{i}{k+1}}+\braket{{\Psi}(\xx{k})}{\ala{k+1}{i}}
}
and such that $\braket{{\Psi}(x^d)}{\ala{k+1}{i}}=\y_i^d$ for each $i\in\pl$. 
\begin{assumption}
 For every $\{\theta_i\}_{i=1}^\n$ where $\theta_i\in \Theta_i$, there exist
 $\alpha_i\in \mb{R}^{\sn_i}$ for
    each $i\in\pl$ such that
    $x^d$ is the estimated collective response---that is,~\eqref{eq:update-1}
    is satisfied for each $i\in\pl$---and such that 
    $\braket{{\Psi}(x^d)}{\ala{k+1}{i}}=\y_i^d$.
\label{ass:myopic}
\end{assumption}

As in the Nash--play case, finding the incentive parameters at each iteration that induce the desired
response amounts to solving a set of linear equations for each player. That is,
for each $i\in\pl$,
the coordinator must solve
\begin{align}
  \label{eq:alpha-eqn-2}
  0_{p_i\times 1}
&= \tilde{\zeta}_i^{k+1}-\tilde{\Lambda}_i^k\ala{k+1}{i}=\bmat{x_i^d-\braket{\Phi(\xx{k})}{\tit{i}{k+1}}\\
  v_i^d}-\bmat{{\Psi}_i(\xx{k})^T\\
  \Psi_i(x^d)^T}\ala{k+1}{i}
 \end{align}
 for $\ala{k+1}{i}$. Define $\mc{A}_m(\tu{k+1}, \xx{k}, x^d, \y^d)$  to be the
set of $\ala{k+1}{i}$ that satisfy~\eqref{eq:alpha-eqn-2}.

The above set of equations will have a solution if the
 matrix $\td{\Lambda}^k_i$
 has rank $2$. Choosing the set of basis functions
 $\{\psi_{j}\}_{j=1}^{\sn}$ such that $\sn>2$ makes this rank condition
 not unreasonable. One unfortunate difference between the Nash--play case and the present
 case of myopic--play is that in the former the coordinator could check the rank
 condition \emph{a priori} given that it does not depend on the observations. On
 the other hand, 
$\td{\Lambda}_{i}^k$ depends on the observation at each iteration and thus, can
only be verify online.
 
As before,  the problem can be cast as a least{\ds}squares
 optimization problem with cost
 $\|\td{\zeta}_i^{k+1}-\td{\Lambda}_i^k\ala{k+1}{i}\|_2^2$.

 %

\section{Convergence in the Noise Free Case}
\label{sec:main}

\begin{algorithm}[t]
  \caption{Adaptive Incentive Design}
  \label{alg:nash}
  \begin{algorithmic}[1]
 
  \item $k\leftarrow 0$
  \item \textbf{If} Myopic--play: receive $x^0$
\item   Choose  $\tit{i}{0}\in \Theta_i$, $\ala{0}{i}\in \mb{R}^{s}$ for each $i\in\pl$ 
    \item Issue incentive mapping with parameters $\ala{0}{i}$

 \item \textbf{do}
 \item$\quad$ issue incentives with parameters $\alpha^k$
  \item $\quad$ Receive $x^{k+1}\in \mc{X}_{k}(x^{(0,k)},\theta^\ast,
    \alpha^k)$.
  \item $\quad$ Compute $y^{k+1}$, $\xi^k$ and incur loss
    $\ell(\tit{i}{k})$
  \item $\quad$ $\tit{i}{k+1}=P_{\tit{i}{k}}^{k+1}\left( \eta_k\nabla
    \ell(\tit{i}{k}) \right)$ for each $i\in\mc{I}$
\item $\quad$ \textbf{for} $i \in \mc{I}$:
\item $\quad\ $ \textbf{if} Nash--play:
  \item $\quad\ \ $ $\ala{k+1}{i}\in\mc{A}(\tu{k+1},
    x^d, \y^d)$
    [or
    $\mc{A}_s(\tu{k+1}, x^d, \y^d)$]
\item $\quad\ $ \textbf{elif} myopic--play: $\ala{k+1}{i}\in\mc{A}_m(\tu{k+1},
    x^d, \y^d)$ 
  \item $\quad$ $k\leftarrow k+1$
 \item \textbf{end do}
  \end{algorithmic}
\end{algorithm}

\label{subsec:converge-myopic}
In Algorithm~\ref{alg:nash}, the steps of the utility learning and incentive
design algorithm are formalized. We now discuss the convergence results for the
proposed
algorithm.

  \begin{definition}
    If, for each $i\in\pl$, there exists a constant
  $0<c_{i,s}<\infty$ such that $\xi_i^k(\xi_i^k)^T\leq c_{i,s}I$ for all $k$,
  then we say the algorithm is {stable}. 
  \end{definition}
  Lipschitz continuity of the functions in $\mc{F}_\phi$ implies stability.
  
  \begin{definition}
    If for each $i\in\pl$, there exists a constant $0<c_{i,p}<\infty$
  such that $c_{i,p}I\leq \xi_i^k(\xi_i^k)^T$ for all $k$, we will say the algorithm is
  {persistently exciting}.   \label{def:persist}
  \end{definition}

  Let $c_{s}=\max_{i\in\pl}c_{i,s}$ and
  $c_{p}=\min_{i\in\pl}c_{i,p}$.
The following lemma is a straghtforward extension of~\cite[Lemma
2.1]{nemirovski:2009aa} and we leave the proof to Appendix~\ref{app:proofs}.
\begin{lemma}
  For every $\ts\in \Theta^{k+1}, \tu{k}\in (\Theta^{k})^\circ$, and $g\in \mb{R}^m$,
  we have 
  \begin{equation*}
 \textstyle   V(P_{\tu{k}}^{k+1}(g), \ts)\leq V(\tu{k},
 \ts)+\braket{g}{\ts-\tu{k}}+\frac{1}{2\nu}\|g\|_\ast^2
    \label{eq:lemma-2}
  \end{equation*}
  \label{lem:tvpmaps}
\end{lemma}
 To compactify the notation, let $V_k(\theta_i)\equiv
 V(\tit{i}{k}, \theta_i)$ and $\Delta
 \tit{i}{k}=\theta_i^\ast-\tit{i}{k}$. 

  \begin{theorem}
Suppose  Algorithm~\ref{alg:nash} with
 prox--mapping defined by $\beta$ (modulus $\nu$)
 is {persistently exciting}, stable and that  
  the step{\ds}size $\eta$
  is chosen such that $\eta-\frac{\eta^2}{2\nu}{c}_s>\vep$ for some $\vep$ such
  that $0<\vep <\frac{1}{2c_p}$ with $c_p=\min_{i\in\pl}c_{i,p}$ and where
   $0<{c}_s<\infty$ is such that $\|\xi_i^k\|_\ast^2\leq {c}_s$.
  We have the following:
  \begin{enumerate}[topsep=2pt, itemsep=0pt]
      \item[(a)] For each $i\in\pl$, $V_k(\ts_i)$ converges and
  \begin{equation}
    \lim_{k\rar\infty}\|(\xi_i^k)^T(\ts_i-\tit{i}{k})\|^2=0.
    \label{eq:limzerou}
  \end{equation}
\item[(b)] If $\beta(\theta_i)=\frac{1}{2}\|\theta_i\|_2^2$, then for each $i\in
  \mc{I}$,  $\tit{i}{k}$ converges {exponentially
  fast} to $\ts_i$.
  \end{enumerate}
     \label{thm:1d-anynorm}
\end{theorem}

\begin{proof}
We prove part (a) first.
Since elements of $\mc{F}_\phi$ are Lipschitz, $\xi_{i}^k(\xi_i^k)^T\leq \hat{c}_sI$
so that we may find
  a ${c}_s>0$ such
  that for all $k$,
  $\|\xi_i^{k}\|_\ast^2\leq {c}_s$.  Lemma~\ref{lem:tvpmaps} implies that
  \begin{equation}
   \textstyle V_{k+1}(\ts_i)\leq
   V_{k}(\ts_i)-\eta\braket{\Delta\tit{i}{k}}{\nabla\ell(\tit{i}{k})}+\frac{\eta^2}{2\nu}\|\nabla\ell(\tit{i}{k})\|_\ast^2.
  \end{equation}
  Hence, 
  \begin{align}
    V_{k+1}(\ts_i)&\leq
    V_{k}(\ts_i)-\eta\|(\xi_i^k)^T(\Delta\tit{i}{k})\|^2\textstyle+\frac{1}{2\nu}\eta^2\|\xi_i^k\|_\ast^2\|(\xi_i^k)^T(\Delta\tit{i}{k})\|^2\notag\\
    &\leq \textstyle
    V_{k}(\ts_i)-\left(\eta-\frac{1}{2\nu}\eta^2\|\xi_i^k\|_\ast^2\right)\|(\xi_i^k)^T(\Delta\tit{i}{k})\|^2\notag\\
    &\leq\textstyle
    V_{k}(\ts_i)-\left(\eta-\frac{1}{2\nu}\eta^2{c}_s\right)\|(\xi_i^k)^T(\Delta\tit{i}{k})\|^2\notag\\
    &\leq V_{k}(\ts_i)-\vep\|(\xi_i^k)^T(\Delta\tit{i}{k})\|^2
  \label{eq:Vs-2}
  \end{align}
  Thus, 
 $\|(\xi_i^k)^T(\Delta\tit{i}{k})\|^2\leq
  \frac{1}{\vep}(V_k(\ts_i)-V_{k+1}(\ts_i))$.

  Summing $k$ from $0$ to $K$, we have
  \begin{align*}
      \textstyle\sum_{k=0}^{K}\|(\xi_i^k)^T(\ts_i-\tit{i}{k})\|^2&\leq
  \vep^{-1}(V_0(\ts_i)-V_{K+1}(\ts_i))\leq \vep^{-1}V_0(\ts_i)
  \end{align*}
  so that $\lim_{K\rar\infty}\sum_{k=0}^{K}\|(\xi_i^k)^T(\ts_i-\tit{i}{k})\|^2\leq
  \vep^{-1}V_0(\ts_i)<\infty$.
 This, in turn, implies that $\lim_{k\rar\infty}
  \|\xi_i^k(\ts_i-\tit{i}{k})\|^2=0$.
From~\eqref{eq:Vs-2} and the fact that $V_k(\ts_i)$ is always postive, we see
that $V_k(\ts_i)$ is a decreasing sequence and hence, it converges.
The analysis holds for each $i\in\pl$. 

Now, we show part (b). Suppose that
$\beta(\theta_i)=\frac{1}{2}\|\theta_i\|_2^2$. Then, starting with the
inequality in \eqref{eq:Vs-2}, we have that
\begin{align*}
    V_{k+1}(\ts_i)&\leq
    V_k(\ts_{i})-\vep(\Delta\tit{i}{k})^T\xi_i^k(\xi_i^k)^T(\Delta\tit{i}{k})\leq V_k(\ts_{i})-\vep c_p(\Delta\tit{i}{k})^T(\Delta\tit{i}{k})\leq V_k(\ts_{i})(1-2c_p\vep)
  \label{eq:partb}
\end{align*}
since $\xi_i^k(\xi_i^k)^T\leq c_{s}I$ (i.e.~stability),
  $\eta-\frac{1}{2}\eta^2c_s>\vep$ by
  construction, and $c_pI\leq \xi_i^k(\xi_i^k)^T$ (i.e.~persistence of excitation).
  Since $0<\vep<1/(2c_p)$, we have that $1-2c_p\vep<e^{-2c_p\vep}$ so that
    $V_{k+1}(\ts_{i})<e^{-2c_p\vep}V_k(\ts_{i})$.
  This implies that
    $V_K(\ts_i)<e^{-2c_pK\vep}V_0(\ts_{i})$.
  Therefore we have that $\tit{i}{k}\rightarrow \ta{i}^\ast$ exponentially fast. The
  same argument holds for each $i\in\pl$.
\end{proof}

For general prox--mappings, Theorem~\ref{thm:1d-anynorm} lets us conclude that 
the observations converge to zero and that the prox{\ds}function $V_k(\theta_i)$
  converges. Knowing the parameter values---a
  consequence of choosing $\beta(\theta_i)=\frac{1}{2}\|\theta_i\|_2^2$---allows
  for the opportunity to gain qualitative insights into how agents' preferences
  affect the outcome of their strategic interaction. 
  On the other hand, more general distance generating functions selected to
  reflect the geometry of $\Theta$ have the potential to
  improve convergence rates~\cite{nemirovski:2009aa}.

\begin{corollary}
  Suppose agents play according to the myopic update
  rule~\eqref{eq:myopic_update}.  Under the assumptions of
  Theorem~\ref{thm:1d-anynorm}, for each $i\in\pl$,
  $\|\bx{i}{t+1}-x_i^d\|^2\rar 0$
and $|\vv{i}{t+1}-v^d|^2\rar 0$. 
 \label{cor:ubdd}
\end{corollary}
\begin{proof}
    Since agents play myopically, we have that
    $x^d_i=\braket{\Phi(\xx{k})}{\tit{i}{k}}+\braket{\Psi(\xx{k})}{\ala{k}{i}}$
and the predicted induced response is
$x^d_i=\braket{\Phi(\xx{k})}{\tit{i}{k}}+\braket{\Psi(\xx{k})}{\ala{k}{i}}$
so that   $\|\bx{i}{k+1}-x_i^d\|^2=\|(\xi_i^k)^T(\ts_i-\tit{i}{k})\|^2$. Thus,
by Theorem~\ref{thm:1d-anynorm}-(a),
$\lim_{k\rar\infty}\|\bx{i}{k+1}-x_i^d\|^2=0$.
Moreover, by Assumption~\ref{ass:myopic}, we know that each 
$\ala{k+1}{i}$ satisfies $\braket{\Psi(x^{d})}{\ala{k+1}{i}}=v_i^d$. Define
$\vv{i}{k+1}=\braket{\Psi(\xx{k+1})}{\al{k+1}}$. Then, 
\begin{align*}
    \notag|\vv{i}{k+1}-v_i^d|^2&=|\braket{\Psi(\xx{k+1})-\Psi(x^d)}{\ala{k+1}{i}}|^2\\
    &\leq
\|\Psi(\xx{k+1})-\Psi(x^d)\|^2\|\ala{k+1}{i}\|_\ast^2.
\end{align*}
Since each element of $\mc{F}_{\psi}$ is Lipschitz, 
  $\|\Psi(\xx{k+1})-\Psi(x^d)\|\leq C\|\xx{k+1}-x^d\|$
  for some constant $C>0$.
  Hence,
  $|\vv{i}{k+1}-v_i^d|^2\leq
  C^2\|\ala{k+1}{i}\|_\ast^2\|\xx{k+1}-x^d\|^2$
and since $\|\xx{k+1}-x^d\|^2$ converges to zero and $\|\ala{k+1}{i}\|_\ast^2<\infty$,
we get that $|\vv{i}{k+1}-v^d|^2$ converges to
  zero. 
\end{proof}



In the {\nplay} case, we can use the fact that non{\ds}degenerate differential Nash equilibria are
structurally stable~\cite[Theorem~3]{ratliff:2015aa}  to
determine a bound on how close an equilibrium of the incentived game,
$(\fg{1}(x; \ts_1), \ldots, \fg{\n}(x; \ts_\n))$
with $\gamma_i(x)=\braket{\Psi(x)}{\ala{k}{i}}$ for each $i\in\pl$,
is to the desired
Nash equilibrium $x^d$.
Note that the observed
Nash equilibrium $\xx{k+1}$ is in the set of Nash equilirbia of the
incentivized game.

Let $\omegag(\theta,x)=(D_1\fg{1}(x;\theta_1), \ldots, D_n\fg{n}(x; \theta))$ be
the \emph{differential game form}~\cite{ratliff:2015aa} of the
incentivized game $\mc{G}_\gamma=(\fg{1}, \ldots, \fg{\n})$.
By a
slight abuse of notation,  we will denote $D_1\omegag(\theta, x)$ and
$D_2\omegag(\theta, x)$
 as the local
   representation of the differential of $\omegag$ with respect to $\theta$ and
   $x$ respectively. If the parameters of the incentive
   mapping $\al{k}$ at
   each iteration $k$ are $C^2$ with respect to $\xx{k}$
   and 
   $\tu{k}$, then the
   differential of $\omega$ is well--defined. We remark that we formulated the
   optimization problem for finding the $\alpha$'s as a constrainted
   least{\ds}squares problem and there are existing results for determining
   when solutions to such problems are continuously dependent on parameter
   perturbations~\cite{lotstedt:1983aa,bonnans:2000aa}. 
  
  \begin{theorem}
      Suppose that for each $k$, $\al{k}(\theta,x)\in C^2(\mb{R}^{nm}\times
    \mb{R}^{n}, \mb{R}^s)$ is chosen such that $x^d$ is a non--degenerate differential Nash equilibrium. For
    $\|\tu{k}-\ts\|$ sufficiently small, there is a Nash equilibrium
    $x^{\ast}$ of
    $\mc{G}_\gamma=(\fg{1}(x; \ts_1), \ldots, \fg{\n}(x; \ts_\n))$
    that is near the desired Nash equilibrium, i.e.~there exists 
    $\bvep>0$, such that for all $\tu{k}\in
    B_{\bvep}(\ts)$, 
    \begin{equation*}\textstyle
     \|x^{\ast}-x^d\|\leq  \left( \sup_{0\leq \lambda\leq 1}\|Dg(
     (1-\lambda)\ts+\lambda\tu{t})\| \right)\|\tu{k}-\ts\|
      \label{eq:bound-vep}
    \end{equation*}
    where
      $Dg(\theta)=-(D_2\omegag)^{-1}(\theta, x^d)\circ D_1\omegag(\theta, x^d)$.
    Furthermore, if $\|Dg(\theta)\|$
 is uniformly bounded by $M>0$ on $B_{\bvep}(\ts)$, then
     $\|x^{\ast}-x^d\|\leq  M\|\theta^{k}-\ts\|$.
    \label{thm:nash-bound}
  \end{theorem}
  \begin{proof}
   Consider the differential game form $\omegag(\theta, x)$ which is given   by 
   \begin{align*}
  \textstyle     \omegag(\theta,
   x)&= \textstyle \sum_{i=1}^{\n}\big(\sum_{j=1}^{m}D_{i}\phi_{j}(x)\theta_{i,j}+ \textstyle \sum_{k=1}^{s}D_{i}\psi_{k}(x)\alpha_{i,j}\big)
   dx_i.
   \end{align*}

   Since $x^d$ is a non--degenerate differential Nash equilibrium,
   $D_2\omegag(\ts, x^d)$ is an isomorphism. Thus, by the Implicit Function
   Theorem~\cite[Theorem 2.5.7]{abraham:1988aa}, there exists a neighborhood
   $W_0$ of $\ts$ and a $C^1$ function
   $g:W_0\rar X$ such that for all $\theta\in W_0$, 
   $\omega(\theta, g(\theta))=0.$
   Furthermore, 
$Dg(\theta)=-(D_2\omegag)^{-1}(\theta, x^d)\circ D_1\omegag(\theta, x^d).$
Let $B_{\bvep}(\ts)$ be the largest $\bvep$--ball inside of $W_0$. 
 Since $B_{\bvep}(\ts)$ is
    convex, by Proposition~\cite[Proposition~2.4.7]{abraham:1988aa}, we have that
    \begin{equation*}\textstyle
      g(\tu{k})-g(\ts)=(\int_0^1
      Dg((1-\lambda)\ts+\lambda\tu{k}  )\ d\lambda)\cdot
      (\tu{k}-\ts)
      \label{eq:g-convex}
    \end{equation*}
    Hence, since $ \|x^{\ast}-x^d\|=\|g(\tu{k})-g(\ts)\|$, we have that
    \begin{align*}
      \|x^{\ast}-x^d\| 
      &\leq \textstyle(
      \sup_{0\leq \lambda\leq 1}\|Dg( (1-\lambda)\ts+\lambda\tu{k}
      )\|
      ) \|\tu{k}-\ts\|
    \end{align*}
    Now, if $\|Dg(\ts)\|$ is uniformly bounded by $M>0$ on $B_{\bvep}(\ts)$,
    then its straightforward to see from the above inequality
    that
      $\|x^{\ast}-x^d\|\leq  M\|\tu{k}-\ts\|.$
  \end{proof}
  
  As a consequence of Theorem~\ref{thm:1d-anynorm}-(b)
  and Theorem~\ref{thm:nash-bound}, there exists a finite iteration $k$ for which
  $\|\ts_i-\tit{i}{k}\|_2^2$ is sufficiently small for each $i\in\pl$ so
  that a local Nash equilibrium of the incentivized game at time $k$ 
  is arbitrarily close to the desired Nash
  equilibrium $x^d$.
  
There may be multiple Nash equilibria of the incentivized
  game; hence, if the agents converge to $x^{\ast}$
  then the observed local Nash
  equilibrium is near the desired Nash equilibria. We know that for stable,
  non--degenerate differential Nash equilibria, agents will converge
  \emph{locally}
  if following the gradient flow determined by the differential game form
  $\omega$~\cite[Proposition~2]{ratliff:2015aa}. 
  \begin{corollary}
    Suppose the assumptions of Theorem~\ref{thm:nash-bound} hold and that
    $x^\ast$ is stable. If agents follow the gradient of
    their cost, i.e.~$-D_if_i$, then they will converge locally to $x^\ast$.
    Moreover, there exists an $\bar{\vep}>0$ such that for
    all $\tu{k}\in B_{\bar{\vep}}(\ts)$, 
    \begin{equation*}    \textstyle
     \|x^{\ast}-x^d\|\leq  \left( \sup_{0\leq \lambda\leq 1}\|Dg(
     (1-\lambda)\ts+\lambda\tu{k})\| \right)\|\tu{k}-\ts\|
    \end{equation*}
  \end{corollary}
  The proof follows directly from Theorem~\ref{thm:nash-bound}
  and~\cite[Proposition~2]{ratliff:2015aa}.
  
    The size of the neighborhood of initial conditions for which agents converge
    to the desired Nash can be
    approximated using techniques for computation of region of attraction via a
    Lyapunov function~\cite[Chapter 5]{sastry:1999aa}. This is in part due to
    the fact that in the case
    where $\al{k}$ is chosen so that $x^d$ is stable, i.e. 
    $d\omegag(\tu{k}, x^d)>0$, we have that $d\omegag(\ts, x^{\ast})>0$ for
    $\tu{k}$ near $\ts$ since the spectrum of $d\omegag$ varies
    continuously.
    
    Moreover, it is possible to explicitly construct the neighborhood $W_0$ obtained via
    the Implicit Function Theorem in Theorem~\ref{thm:nash-bound} (see,
    e.g.~\cite[Theorem 2.9.10]{hubbard:1998aa} or \cite{ratliff:2015ab}). 

The result of Theorem~\ref{thm:1d-anynorm}-(b) implies that the incentive value
under $x^\ast$
from Theorem~\ref{thm:nash-bound}
 is arbitrarily close to the desired incentive value.
  \begin{corollary}
 Under the assumptions of Theorem~\ref{thm:1d-anynorm}-(b) 
  and Theorem~\ref{thm:nash-bound}, there exists a finite $K$ 
      such that for all $k\geq K$, 
      $\|x^{\ast}-x^d\|^2_2\leq M\bar{C}e^{-2c_pk\vep}$
      where $\bar{C}=\n\max_i\{2V_0( \ts_i)\}$, $c_p=\min_{i\in\mc{I}}c_{i,p}$, and $x^\ast$ is the (local) Nash equilibrium of
     the incentivized game satisfying 
     $\|x^\ast-x^d\|\leq M\|\tu{k}-\ts\|$ for all $\tu{k}\in
     B_{\bar{\vep}}(\ts)$.
     Furthermore, for each $i\in\pl$,
     $|v_i^{\ast}-v_i^d|^2\leq MC^2\|\ala{k}{i}\|_2^2\bar{C}e^{-2c_pk\vep}$
     for all $k\geq K$
     where $C$ is the Lipschitz bound on $\Psi$ and
     $v_i^{\ast}=\braket{\Psi(x^\ast)}{\ala{k}{i}}$. 
   \end{corollary}
   \begin{proof}
      Choose $K$ such that, for each $i\in\pl$,
     $2V_0( \ts_i)e^{-2\vep c_p K}<\bar{\vep}$ so that 
     $\|\tit{i}{k}-\theta_i\|_2^2\leq 2V_0( \ts_i)e^{-2\vep
     c_p k}$ for all $k\geq K$. 
     Thus, $\|\tu{k}-\ts\|_2^2\leq Ce^{-2\vep
     c_p k}$, for all $k\geq K$. By Theorem~\ref{thm:nash-bound}, we have that
     $\|x^{\ast}-x^d\|_2^2\leq M\bar{C}e^{-2\vep
c_p k}$ for all $k\geq K$ 
where $M$ is the uniform bound on $\|Dg(\ts)\|$.   We know that
$v_i^d=\braket{\Psi(x^{d})}{\ala{k}{i}}$ and
$v_i^{\ast}=\braket{\Psi(x^{\ast})}{\ala{k}{i}}$. Hence, since each element of
  $\mc{F}_\psi$ is Lipschitz, we have that for all $k\geq K$,
 that
 \begin{align*}
     |v_i^{\ast}-v_i^d|^2&=|\braket{\Psi(x^{\ast})-\Psi(x^d)}{\ala{k}{i}}|^2\\
  & \leq \|\Psi(x^{\ast})-\Psi(x^d)\|\|\ala{k}{i}\|_2^2\\
  & \leq C^2\|\ala{k}{i}\|_2^2\|x^{\ast}-x^d\|_2^2\\
  & \leq MC^2\|\ala{k}{i}\|_2^2\bar{C}e^{-2\vep
    c_p k}. 
   \label{eq:vbd-nashi}
 \end{align*}
   \end{proof}

    We have argued that following Algorithm~\ref{alg:nash} with a particular
    choice of prox{\ds}mapping, the parameter estimates of each $\ts_i$
    converge to the true values and as a consequence we can characterize the
    bound on how close the observed response and incentive value are to their
    desired values. Knowing the true parameter values (even if obtained
    asymptotically) for $\ts$ allows the
    {\coordinator} to
  make qualitative insights into the rationale behind the observed responses.

\section{Convergence in the Presence of Noise}
\label{sec:noise}
\label{sec:noise}
In this section, we will use the unified framework that describes both the case
where the agents play according to Nash and where the agents play myopically.
  However, we consider \emph{noisy} updates given by
\begin{equation}
    y_i^{k+1}=\braket{\xi_i^k}{\ts_i} + w_i^{k+1}
  \label{eq:update-n}
\end{equation}
for each $i\in\pl$ where $w_i^{k+1}$ is an indepdent, identically distributed (i.i.d) 
real stochastic process defined on a probability space $(\Omega, \mc{F}, \PP)$ adapted
to the sequence of increasing sub-$\sigma$--algebras $(\mc{F}_k, k\in \mb{N})$, where
$\F_k$ is the $\sigma$--algebra generated by the set $\{y_i^t,
\ala{t}{i},
w_i^t, t\leq k\}$ and
such that the following hold\footnote{We use the abbreviation a.s.~for
    \emph{almost surely}.}:
\begin{subequations}
\begin{align}
  &\mb{E}[w_i^{k+1}| \F_k]=0\  \forall k \label{eq:zeromean-0}\\
  &\mb{E}[(w_i^{k+1})^2| \F_k]=\sigma^2>0 \  \text{a.s.}\  \forall k
  \label{eq:zeromean-1}\\
  \textstyle \sup_k& \mb{E}[(w_i^{k+1})^4| \F_k]<+\infty\ \text{a.s.}
  \label{eq:zeromean-2}
\end{align}
\label{eq:zeromean}
\end{subequations}
 Note that $\F_k$ is also the $\sigma$--algebra generated by
$\{y_i^t, \xi_i^t, t\leq k\}$ since $w_i^k$ can be deduced from $y_i^k$ and
$\xi_i^{k-1}$ through the relationship
$w_i^{k}=y_i^k-\braket{\xi_i^{k-1}}{\ts_i}$~\cite{kumar:1986aa}.

\begin{theorem}
  Suppose that for
  each $i\in\pl$, 
  $\{w_i^{k}\}$ satisfies \eqref{eq:zeromean-0} and \eqref{eq:zeromean-1} and that Algorithm~\ref{alg:nash},
  with prox--mapping $P_\theta$ associated with $\beta$ (modulus $\nu$), is 
  {persistently exciting} and stable. Let the 
  step--size $\eta_k$
  be selected such that 
     $\sum_{k=1}^\infty\eta_k^2<\infty$ and
  $\eta_k-\frac{\eta_k^2}{2\nu}\tilde{c}_s>0$
  where
   $0<\tilde{c}_s<\infty$ is such that $\|\xi_i^k\|_\ast^2\leq \tilde{c}_s$.
  Then, for each $i\in\pl$, 
  $V_k(\ts_i)$ converges
  a.s.
  Furthermore, if 
  the sequence
  $\{r_k\}$ 
  where $r_k=(\eta_k)^{-1}$ is a non-decreasing,
  non-negative sequence such that $r_k$ is $\F_k$ measurable and there
  exists constants $0<K_1, K_2<\infty$ and $0<{T}<\infty$ such that
  \begin{equation}\textstyle
    \frac{1}{k}r_{k-1}\leq
    K_1+\frac{K_2}{k}\sum_{t=0}^{k-1}\|y_i^{t+1}-\braket{\xi_i^{t}}{\tit{i}{t}
}-w_i^{t+1}\|^2,
  \ \ \forall  \ k\geq {T},
  \label{eq:rbound-22}
\end{equation}
then 
\begin{equation}\textstyle
  \lim_{k\rar\infty}
  \frac{1}{k}\sum_{t=0}^{k-1}\mb{E}[\|w_i^{t+1}-\braket{\xi_i^{t}}{\tit{i}{t}-\theta_i^\ast}\|^2|
\F_t]=\sigma^2\ \ \text{a.s.}
\label{eq:errorcon-2}
\end{equation}
In addition, if \eqref{eq:zeromean-2} holds, then

\begin{equation}\textstyle
  \lim_{k\rar\infty}
  \frac{1}{k}\sum_{t=0}^{k-1}\|w_i^{t+1}-\braket{\xi_i^{t}}{\tit{i}{t}-\theta_i^\ast}\|^2=\sigma^2\ \ \text{a.s.}
\label{eq:errorcon-22}
\end{equation}
  \label{thm:noise-p}
\end{theorem}
The proof follows a similar technique to that presented in~\cite[Chapter
13.4]{kumar:1986aa}; hence, we leave it to Appendix~\ref{app:proofs}. 
 We remark that if $\beta(\theta_i)=\frac{1}{2}\|\theta_i\|_2^2$, then  Theorem~\ref{thm:noise-p}
 implies  $V_k(\ts_i)=\frac{1}{2}\|\ts_i-\tit{i}{k}\|_2^2$ converges a.s.
 \begin{corollary}
   Suppose agents play according to the myopic update
   rule~\eqref{eq:myopic_update} and that the assumptions  of
   Theorem~\ref{thm:noise-p} hold. Then, 
   \begin{align*}\textstyle
\lim_{k\rar\infty}\frac{1}{k}\sum_{t=0}^{k-1}\mb{E}[\|\bx{i}{t+1}+w_i^{t+1}-x_i^d\|^2|
\F_t]
=\sigma^2\ \ \text{a.s.}
\label{eq:cor-res}
\end{align*}
   \label{cor:noise-1}
 \end{corollary}
 \begin{proof}
   Agent $i$'s response is
   $x_i^{k+1}=\braket{\xi_i^k}{\theta_i^\ast}+\braket{\Psi(x^k)}{\alpha^k}$ and the
   {\coordinator} designs $\alpha_i^k$ to satisfy
   $x_i^d=\braket{\xi_i^k}{\theta_i^{k}}+\braket{\Psi(x^k)}{\alpha^k}$. Hence, replacing
   $\braket{\xi_i^t}{\theta_i^t-\theta_i^\ast}$ in \eqref{eq:errorcon-2} completes
   the proof.
 \end{proof}
The results of Theorem~\ref{thm:noise-p}  imply  that the average mean square error
  between 
  the observations and the predictions converges to $\sigma^2$ a.s.~and,
  if we recall, the observations are derived from noisy versions of the
  first--order conditions for Nash. Indeed,
we have shown that
\begin{align*}
    \textstyle\lim_{k\rar\infty}
    \frac{1}{k}\sum_{t=0}^{k-1}&\textstyle\mb{E}[\|\braket{D_i\Phi(\xx{t+1})}{\ts_i}+\braket{\D_i\Psi(\xx{t+1})}{\ala{t}{i}}+w_i^{t+1}\|^2|
  \F_t]=\sigma^2\ \ \text{a.s.}
\end{align*}
or, equivalently,
\begin{align*}
    \textstyle\lim_{k\rar\infty}
    \frac{1}{k}\sum_{t=0}^{k-1}&\textstyle\mb{E}[\|\braket{D_i\Phi(\xx{t+1})}{\ts_i-\tit{i}{t}}+w_i^{t+1}\|^2|
\F_t]=\sigma^2\ \ \text{a.s.}
\end{align*}

  On the other hand, in the Nash--play case, it is
  difficult to say much about the observed Nash equilibrium except in
  expectation. In particular, we can consider a modified version of
  Theorem~\ref{thm:nash-bound} where we consider the differential game form in
  expectation---i.e.~at iteration $k$,  the
  differential game form for the induced game is 
  \begin{align*}
      \textstyle
  \tdomegag(\theta,
  x)&=\textstyle\sum_{i=1}^\n\mb{E}\big[\braket{D_i\Phi(x)}{\theta_i}+\braket{D_i\Psi(x)}{\ala{k}{i}}+w_i^{k+1}|
  \F_{k-1}\big].
  \end{align*}
\begin{proposition}
  Suppose that $D_2\tdomegag(\theta, x^d)$ is an isomorphism. Then there exists an
  $\epsilon>0$, such that for all $\tu{k}\in B_\epsilon(\ts)$,
   \begin{equation*}
  \textstyle   \|x^{\ast}-x^d\|\leq ( \sup_{0\leq \lambda\leq 1}\|Dg(
     (1-\lambda)\ts+\lambda\tu{k})\| )\|\tu{k}-\theta\|
      \label{eq:bound-vep-noise}
    \end{equation*}
    where
      $Dg(\ts)=-(D_2\tdomegag)^{-1}(\ts, x^d)\circ D_1\tdomegag(\ts, x^d)$.
      and $x^\ast$ is a (local) Nash equilibrium of the incentivized game
      $\mc{G}_\gamma=(\fg{1}(x;
    \ts_1), \ldots, \fg{\n}(x, \ts_\n))$ with
    $\gamma_i(x)=\braket{\Psi(x)}{\ala{k}{i}}$ for each $i\in\pl$. 
Furthermore, 
    if $\|Dg(\ts)\|$
 is uniformly bounded by $M>0$ on $B_{\epsilon}(\ts)$, then
    \begin{equation*}
     \|x^{\ast}-x^d\|\leq  M\|\tu{k}-\ts\|.
      \label{eq:bound-vep-n}
    \end{equation*}
    \label{prop:noisynashbdd}
\end{proposition}

To apply Proposition~\ref{prop:noisynashbdd}, we need a result ensuring that the parameter estimate
$\theta^k$ converges to the true parameter value $\theta^\ast$. One of the
consequences of Theorem~\ref{thm:noise-p}
is that $V_k(\theta_i^\ast)$ converges a.s.~and when
$\beta(\theta_i)=\frac{1}{2}\|\theta_i\|_2^2$,
 $\|\theta^\ast_i -\theta_i^k\|_2^2$ converges a.s. If it is the case that it
 converges a.s.~to a value less than $\vep$, then
 Proposition~\ref{prop:noisynashbdd} would guarantee that a local Nash
 equilibrium of the incentivized game is near the desired non–degenerate
 differential Nash equilibrium in expectation. We leave further exploration of
 the convergence of the parameter estimate $\theta_i^k$ as future work.

\section{Numerical Examples}
\label{sec:examples}

In this section we present several examples to illustrate the theoretical
results of the previous sections\footnote{Code for examples can be found at {\tt github.com/fiezt/Online-Utility-Learning-Incentive-Design}.}.

\subsection{Two-Player Coupled Oscillator Game}
\label{sec:oscillator}
  \begin{figure*}[t]
   \begin{center}
      \subfloat[nominal\label{fig:nash1}]{
          \includegraphics[width=0.265\textwidth]{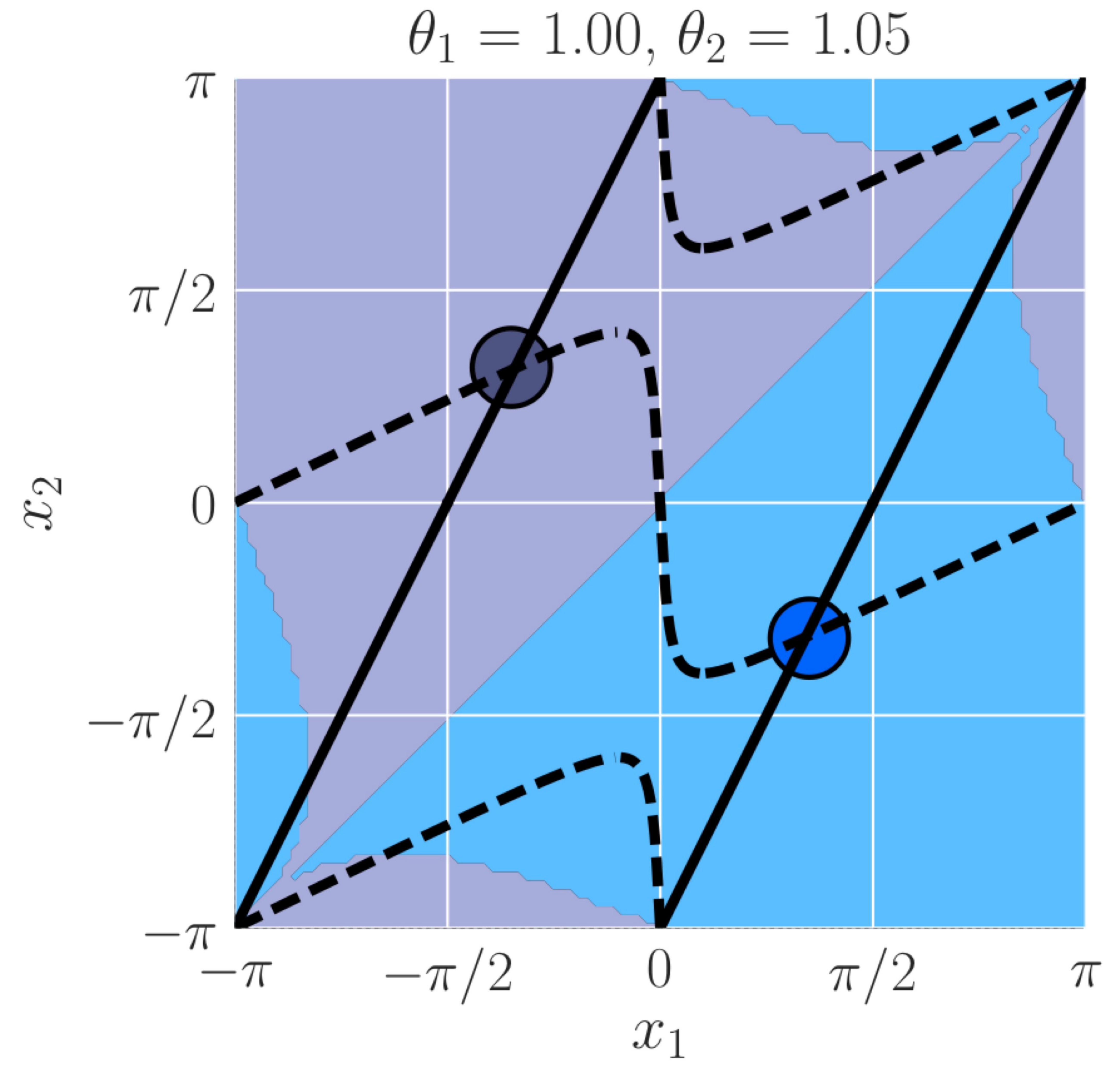}} 
      \hfill\subfloat[incentivized, $\lambda=0.0$ \label{fig:nash2}]{
          \includegraphics[width=0.265\textwidth]{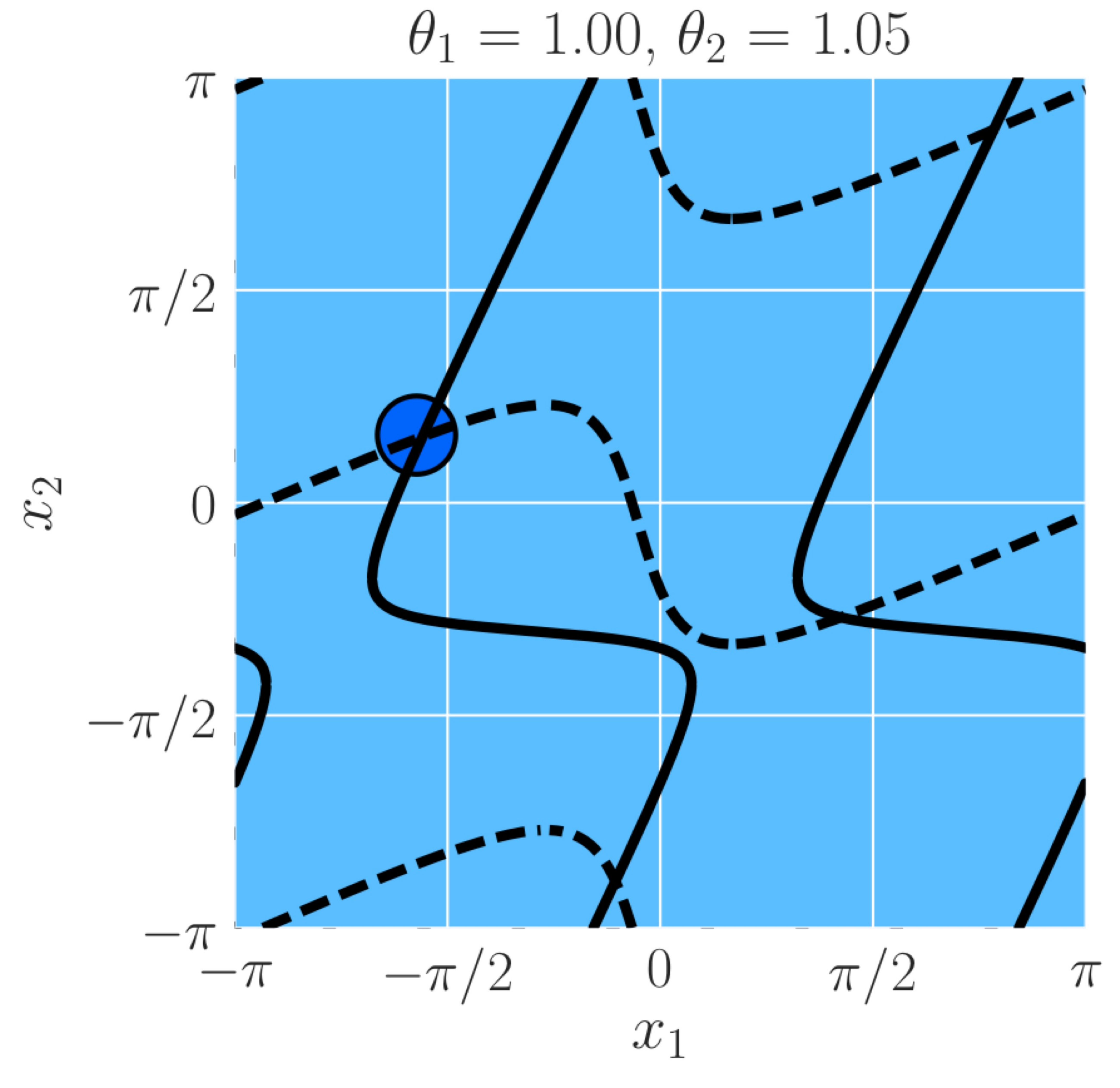}} 
      \hfill\subfloat[incentivized, $\lambda=0.1$\label{fig:nash3}]{
          \includegraphics[width=0.385\textwidth]{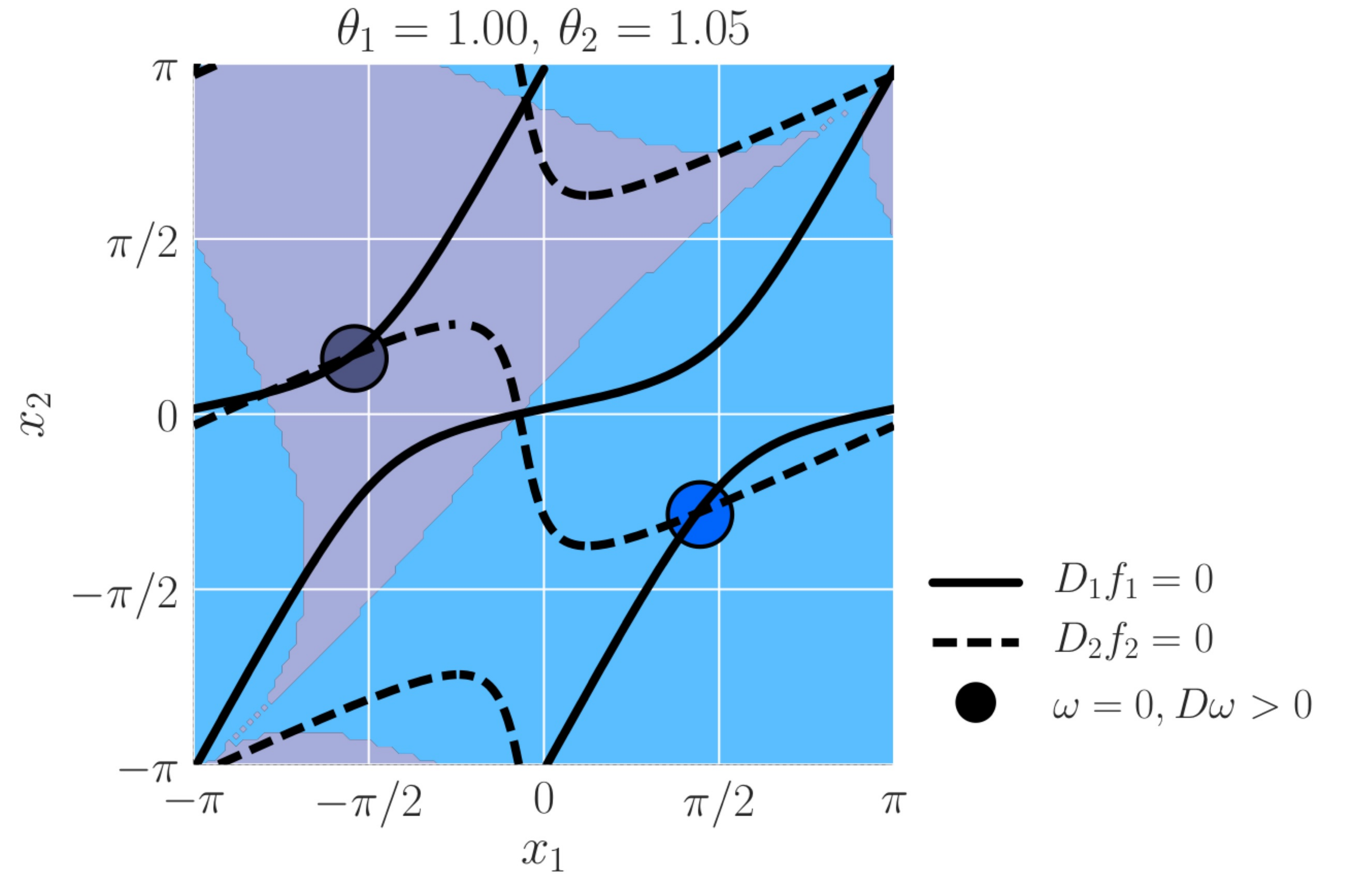}} 
   \end{center}
   \caption{\footnotesize Visualization of Nash equilibria for the (a) nominal,
   (b) incentivized (without regularization), and (c) incentivized (with
   regularization $\lambda=0.1$) games. Curves indicate the set of points where $D_if_i\equiv 0$
   for each $i\in\{1,2\}$; we have $\omega=(D_1f_1,D_2f_2)=0$ wherever the
   curves intersect. The stable, differential Nash
   equilibria are indicated by a dark circle of which there are two for the
   nominal game. Given an initial condition, Nash is calculated 
   via a steepest
   descent algorithm~\cite{ratliff:2015aa}. The empirical basin of attraction
   for each Nash equilibrium is illustrated by the filled region containing the
   point.
   In the incentivized game, the coordinator aims to 
   drive the agents to $(x_1^d,x_2^d)=(-1.8, 0.5)$. For $\lambda=0.0$,
   the
   unique Nash equilibrium of the incentivized game with parameters 
   $(\alpha_1,\alpha_2)=(0.16,-0.94,0.29,-0.11)$ is $(x_1^d,x_2^d)$; however,
   for the regularized incentivized game using with parameters 
   $(\alpha_1,\alpha_2)=(0.13,0,0.15,0)$, there are two stable Nash
   equilibria at $(1.4,-0.9)$ and $(-1.8,0.5)$. 
   }
   \label{fig:nashplots}
 \end{figure*}

\begin{figure*}
    \begin{center}
    \subfloat[][]{\includegraphics[width=0.2375\textwidth]{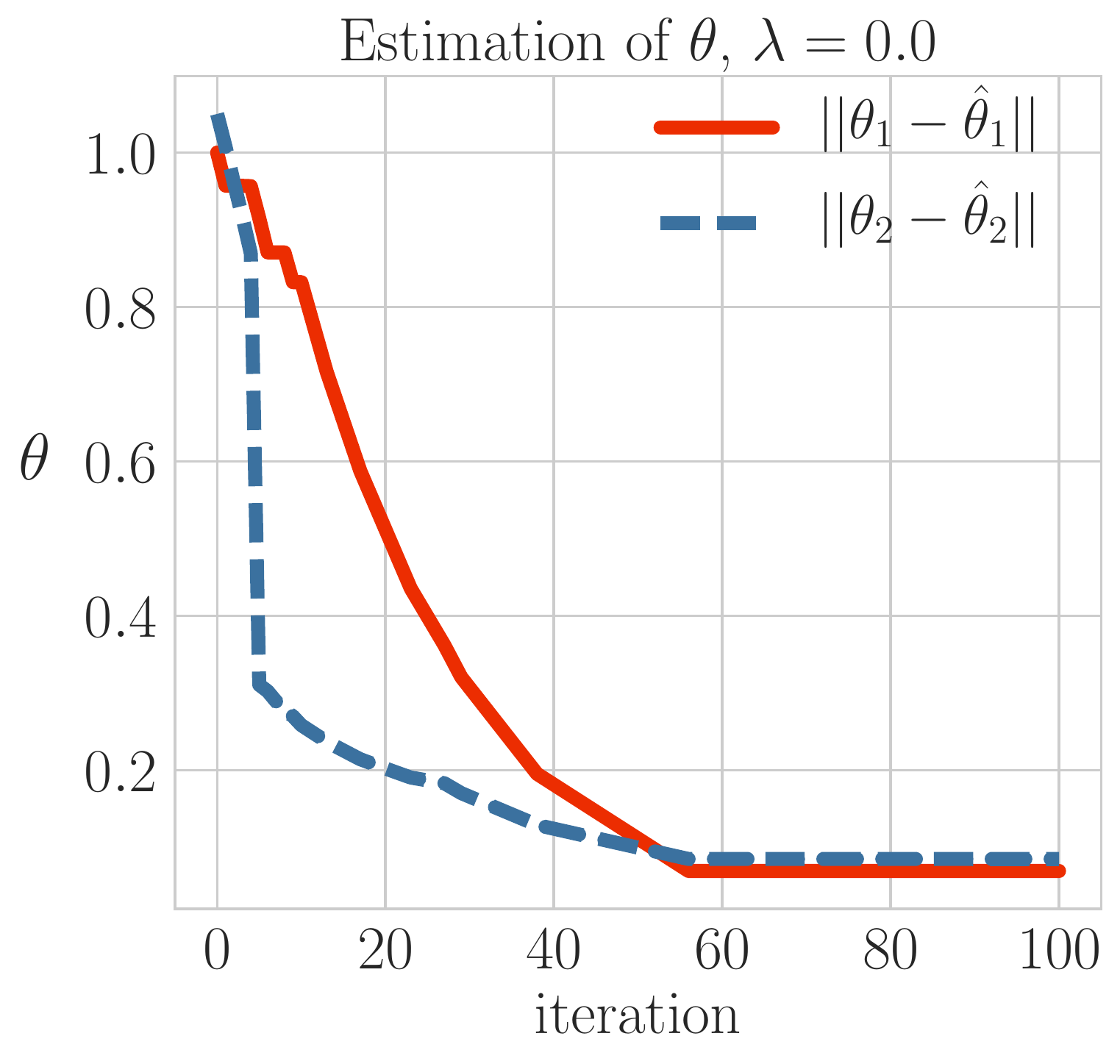}}\hfill
    \subfloat[][]{\includegraphics[width=0.254\textwidth]{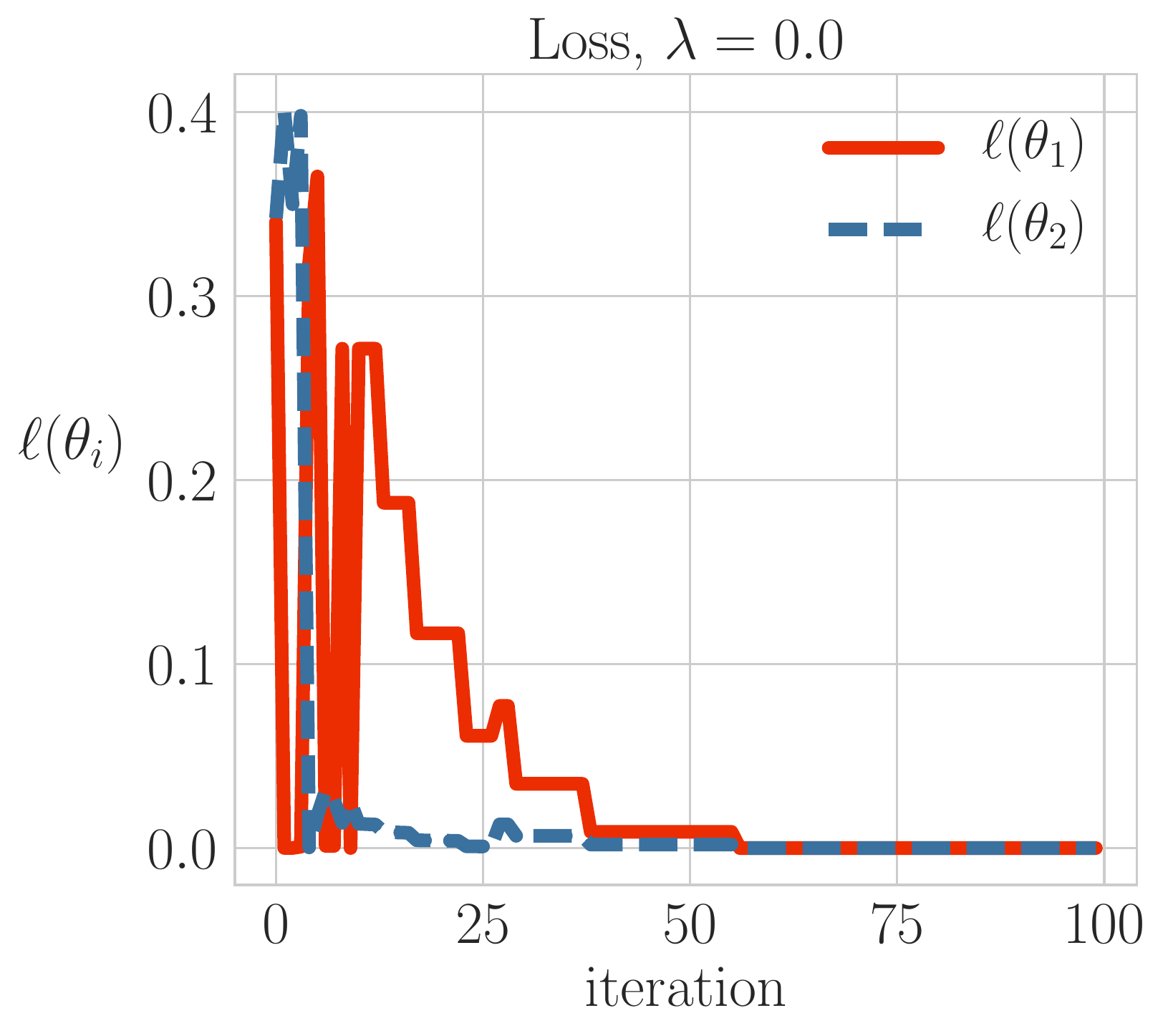}}
\hfill\subfloat[][]{\label{fig:penoreg}\includegraphics[width=0.26\textwidth]{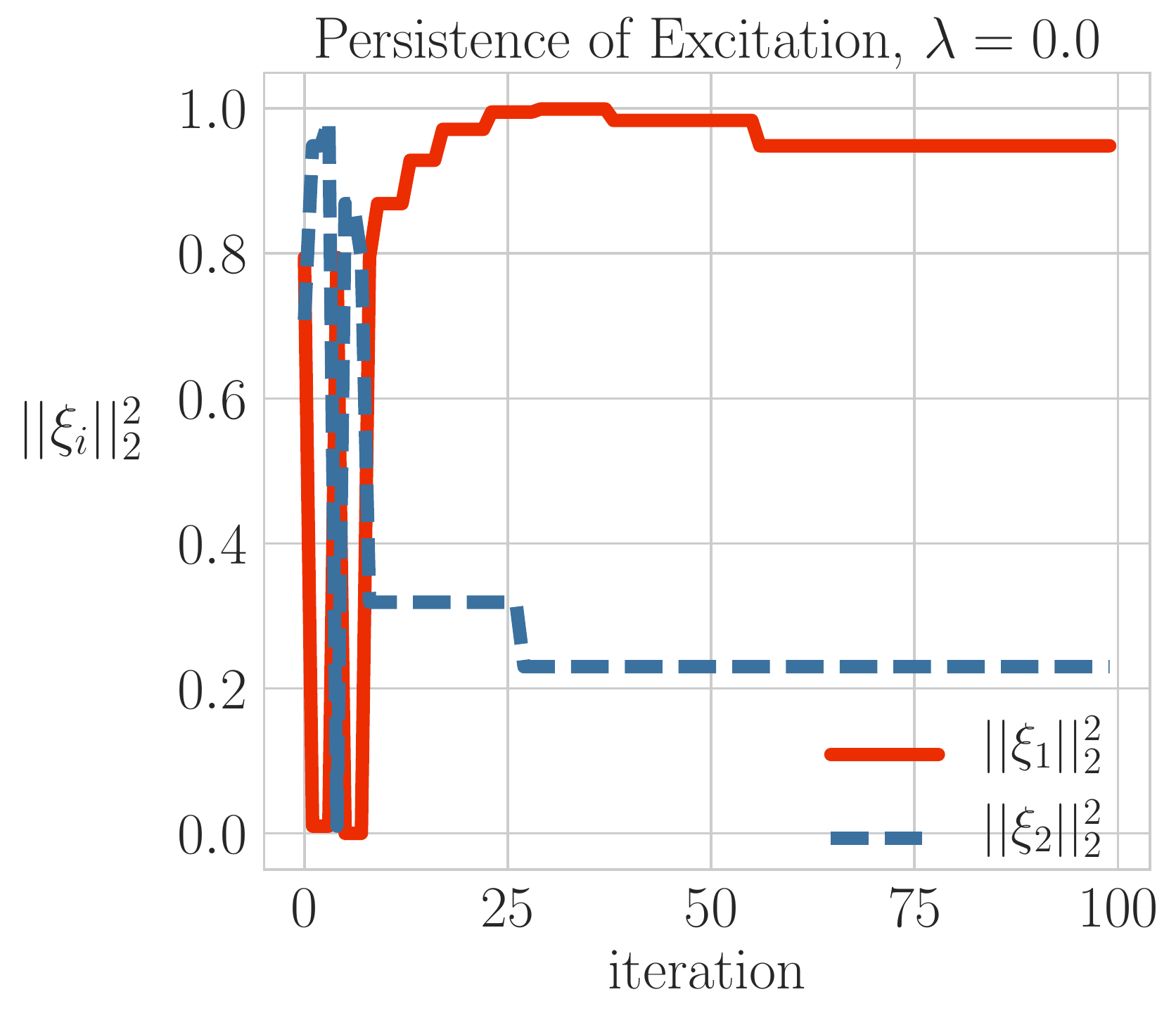}}
       
    \subfloat[][]{\includegraphics[width=0.235\textwidth]{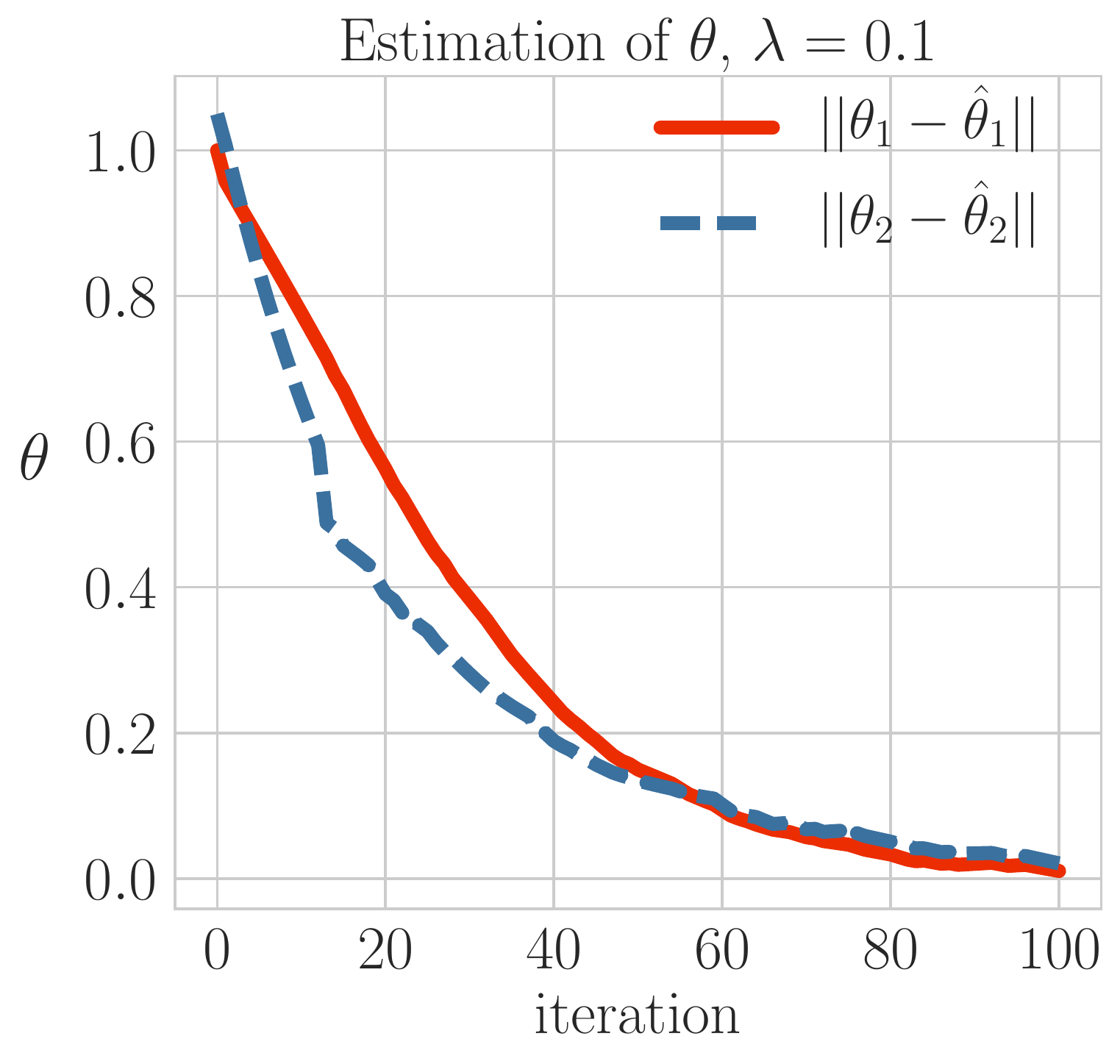}}
                \hfill\subfloat[][]{\includegraphics[width=0.254\textwidth]{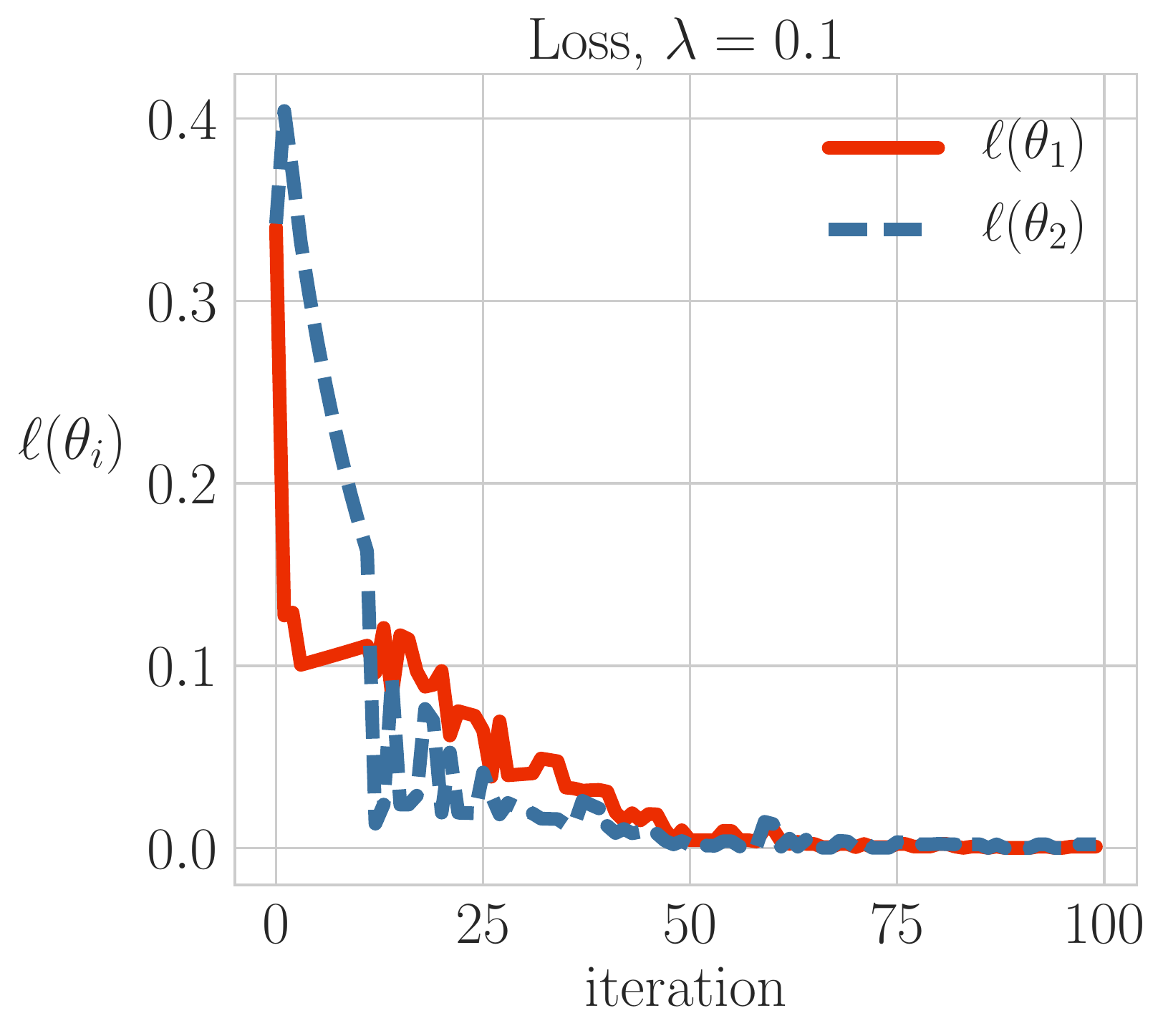}}
\hfill\subfloat[][]{\label{fig:pereg}\includegraphics[width=0.26\textwidth]{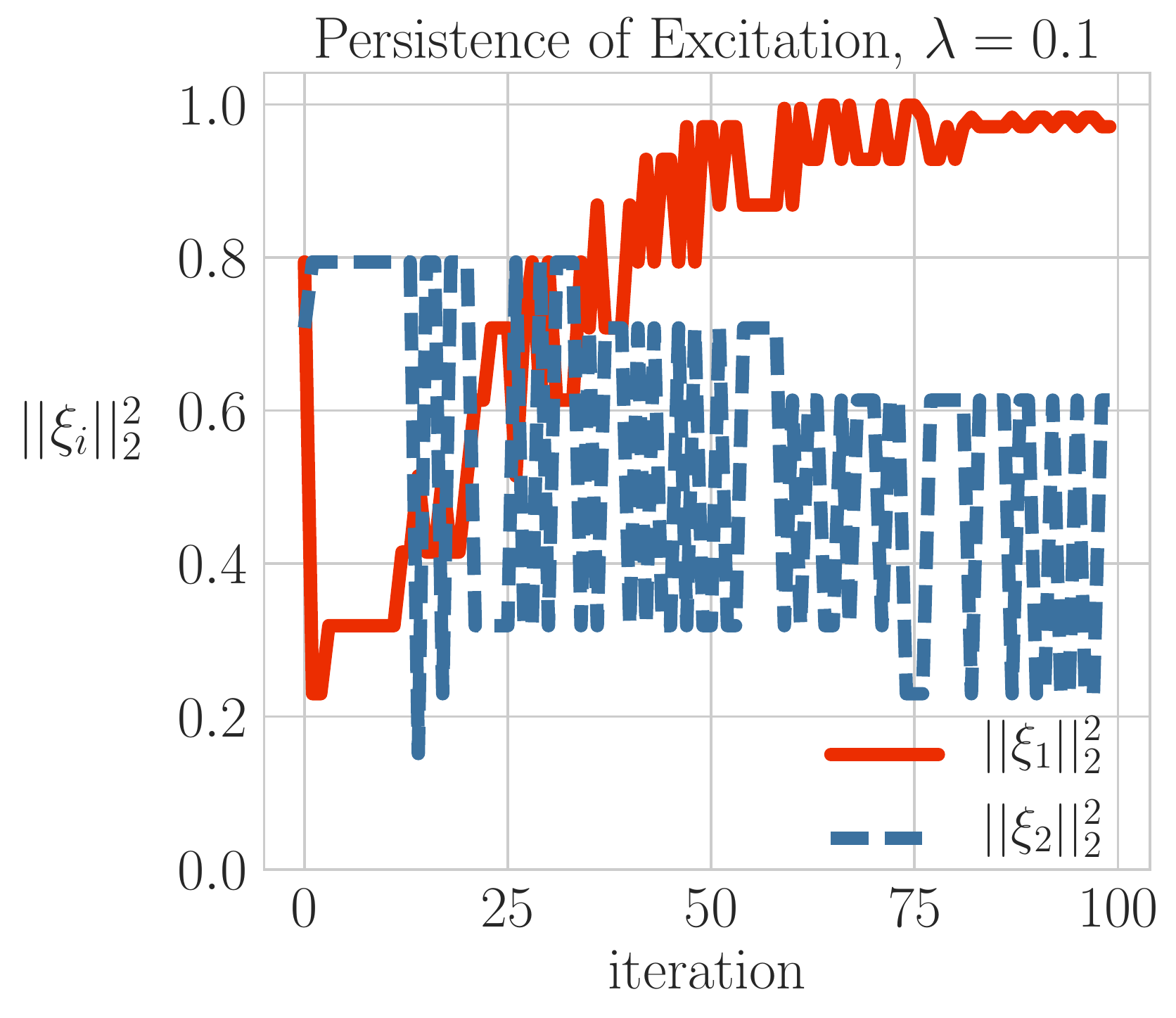}}

    \end{center}
\caption{    \footnotesize  Estimation error for (a) $\lambda=0.0$ and (d)
    $\lambda=0.1$. Loss for (b) $\lambda=0.0$ and (e)
    $\lambda=0.1$. Persistence
    of excitation (as measured by $\|\xi\|^2$) for the coupled oscillator with
    regularization (c) $\lambda=0.0$ and (f) $\lambda=0.1$ on the incentive
parameters.}
            \label{fig:losses}
\end{figure*}

The first example we consider is a game between two players trying to control
oscillators that are coupled. 
Coupled oscillator models are used widely for applications including
power~\cite{dorfler:2012aa}, traffic~\cite{
coogan:2015aa}, and biological~\cite{wang:2015aa}
networks, 
and in coordinated motion
control~\cite{paley:2007aa} among many others. 
Furthermore, it is often the case that coupled oscillators 
 are viewed in a game--theoretic context in order to gain further insight into the system
 properties~\cite{yin:2012aa, goto:2010aa}.
While the example is simple, it demonstrates that even in complex games on
non--convex strategy spaces (in particular, smooth manifolds),
Algorithm~\ref{alg:nash} still achieves the desired result.

Consider a game of coupled oscillators in which we have two players each one
aiming to minimize their nominal cost
\begin{equation}
  f_i(x_i, x_{-i})=-\theta_i^\ast\cos(x_i)+\cos(x_i-x_{-i}).
  \label{eq:ex_cost}
\end{equation}
Essentially, the game takes the form of a location game in which each player wants to be near the origin while
also being as far away as possible from the other player's phase. The coordinator
selects $\alpha_i\in \mb{R}^{2}$, $i\in\{1,2\}$ such that $(x_1^d, x_2^d)$ is
the Nash equilibrium of the game with incentivized costs
\begin{equation}
  f_i^\gamma(x_i,
  x_{-i})=f_i(x_i,x_{-i})+\alpha_{i,1}\sin (x_i-x_i^d)+\alpha_{i,2}\cos(x_i-x_i^d). 
  \label{eq:inc_cost}
\end{equation}
Using Algorithm~\ref{alg:nash}, we estimate $\theta_1$ and $\theta_2$ while
designing incentives that induce the players to play $(x_1^d, x_2^d)$. In this
example, we set  $\theta_1=1.0$ and $\theta_2=1.05$ and set the
desired Nash at $(x_1^d,x_2^d)=(1.8,-0.5)$. 

In Fig.~\ref{fig:nash1}, we show a visualization of the Nash equilibria for the
game where the square represents the unfolded torus with player 1's choice
along the x-axis and player 2's choice along the y-axis. The solid and dashed
black lines depict the zero sets for the derivatives of player 1 and 2
respectively. The intersections of these lines are all the candidate Nash
equilibrium (all the Nash equilibria of the game must be in the set of
intersection points of these lines). There are four intersection points:
$\{x|\omega(x)=0\}=\{(0,\pi), (\pi, \pi), (1.1, -1.0), (-1.1, 1.0)\}$. Of the
four, the first two are not Nash equilibria which can be verified by checking
the necessary conditions for Nash. However, $(0,\pi)$ is a saddle point and
$(\pi,\pi)$ is an unstable equilibrium for the dynamics $\dot{x}=-\omega(x)$.
The last two points are stable differential Nash equilibria. We indicate each of
them with a circle in Fig.~\ref{fig:nash1} and the colored regions show the
empirical basin of attraction for each of the points.

In Fig.~\ref{fig:nash2} and~\ref{fig:nash3}, we show similar visualizations of
the Nash equilibria but for the incentivized game. In the incentive design step
of the algorithm, we can add a regularization term to the constrained least
squares for selecting the incentive parameters. The regularization changes the
problem and may be added if it is desirable for the incentives to be small or
sparse for instance. Fig.~\ref{fig:nash2} and~\ref{fig:nash3} are for the case
without and with  regularization, respectively, where we regularize with $\lambda\|\alpha\|_2^2$. 

As we adjust the regularization parameter $\lambda$, it has the effect of
reducing the norm of the incentive parameters and as a consequence, what happens
is that the number of stable Nash equilibria changes. In particular, a larger
value of $\lambda$ results in small incentive values and, in turn,
non-uniqueness in the number of stable Nash equilibria in the incentivized game. 
There is a tradeoff between uniqueness of the desired Nash equilibrium and how
much the coordinator is willing to expend to incentivize the agents. 

Moreover, in Fig.~\ref{fig:losses}, we see that in
the unregularized case ($\lambda=0$) the estimation error for $\theta_2$ appears
not to converge to zero
while in the regularized case ($\lambda=0.1$) it does. This
is likely 
due to the fact that Algorithm~\ref{alg:nash} is not persistently exciting for
the unregularized problem (see Fig.~\ref{fig:penoreg} where
$\|\xi_1\|_2$ drops to zero) while it is in the regularized case (see
Fig.~\ref{fig:pereg} where $\|\xi_i\|_2$, $i=1,2$ are always non-zero).

We remark that this example highlights that the technical conditions of the theoretical results in the
preceding sections---in particular, persistence of excitation---matter in practice. Indeed, regularization allows for the problem to
be persistently exciting. There are other methods for inducing persistence of
excitation such as adding noise or ensuring the input is \emph{sufficiently
rich}~\cite{boyd:1986aa}.

\subsection{Bertrand Nash Competition}
\label{sec:bertrand}

In this section, we explore classical Bertrand
competition. Following~\cite{bertsimas:2015aa}, we adopt
a stylized model inspired by~\cite{berry:1994aa}. We consider a two firm
competition. Firms are assumed to have full information and they choose their
prices $(x_1,x_2)$ to maximize their revenue
$R_i(x_i,x_{-i},\tau)=x_iF_i(x_i,x_{-i},\tau)$ where $F_i$ is the demand and
$\tau$
is a normally distributed i.i.d.~random variable that is common knowledge and
represents economic indicators such as gross domestic product.

When firms compete with their nominal revenue functions $R_i$, the random
variable $\tau$ changes over time causing \emph{demand shocks} which, in turn,
cause prices to shift. In our framework, we introduce a \emph{regulator} that
wants to push the prices chosen by firms towards some desired pricing level by
introducing \emph{incentives}. It is assumed that the regulator also has
knowledge of $\tau$.

We explore several variants of this problem to highlight the theoretical results as well as desirable empirical results we observe under relaxed assumptions on the knowledge the planner has of agent behavior. 

As noted in the previous example, the assumption of persistence of excitation, which also appears in
the adaptive control literature, is hard to verify \emph{a priori}; however,
selecting basis functions from certain classes such as radial basis functions
which are known to be persistently
exciting~\cite{Kurdila:1995ab,Gorinevsky:1995aa} can help to overcome this
issue. We note that this is not a perfect solution; it
may make the algorithm persistently exciting, yet result in
a solution which is not asymptotically incentive compatible, i.e.~$x_i^k \nrightarrow x_i^d$.
Since we seek classes of problems that are persistently exciting and asymptotically
incentive compatible, this exposes an interesting avenue for future research.

\begin{figure}
    \begin{center}
    \hfill\subfloat[][]{\label{fig:true_gp_responses}\includegraphics[width=0.28\textwidth]{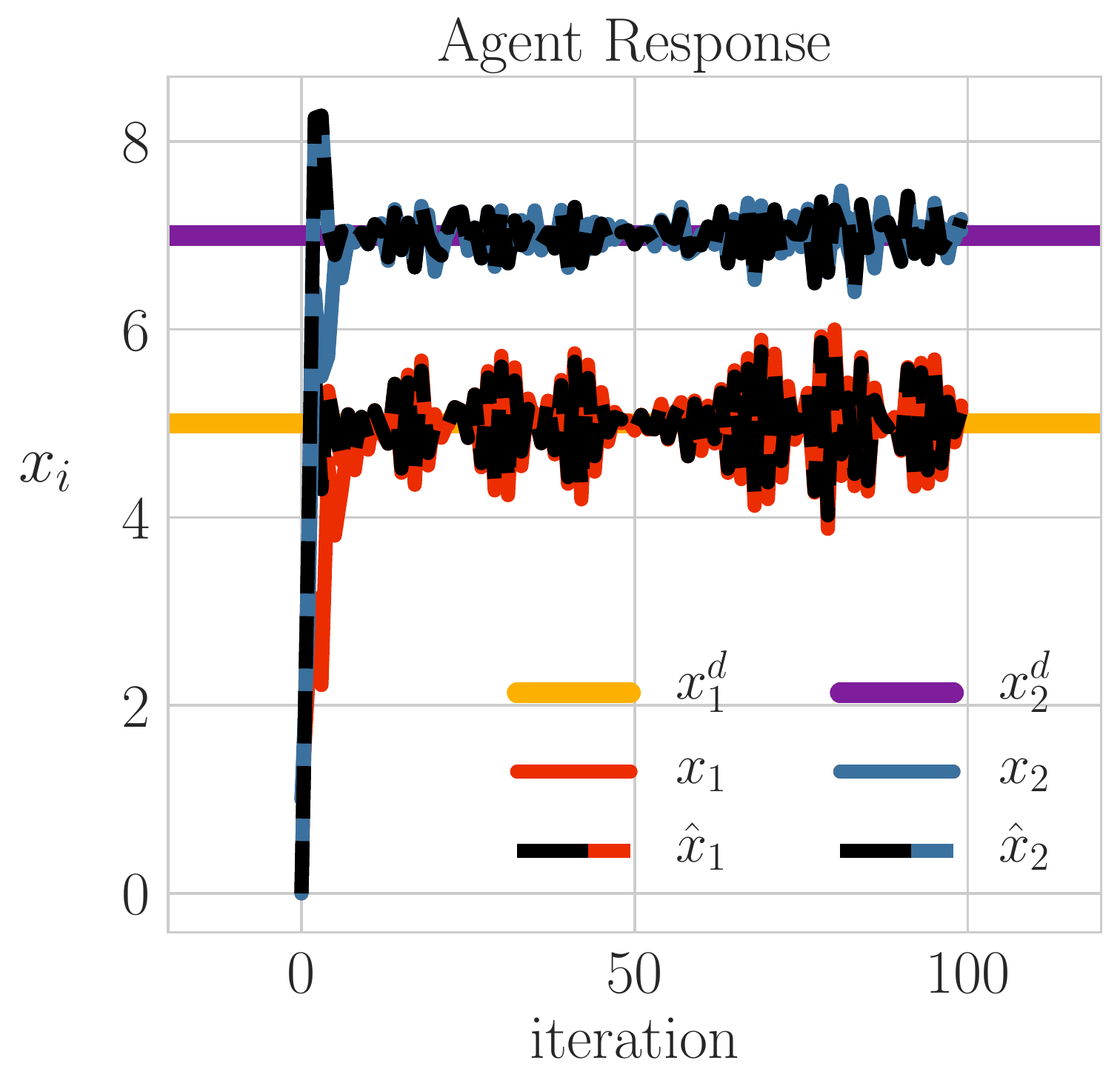}}\hfill\subfloat[][]{\label{fig:true_gp_error}\includegraphics[width=0.29\textwidth]{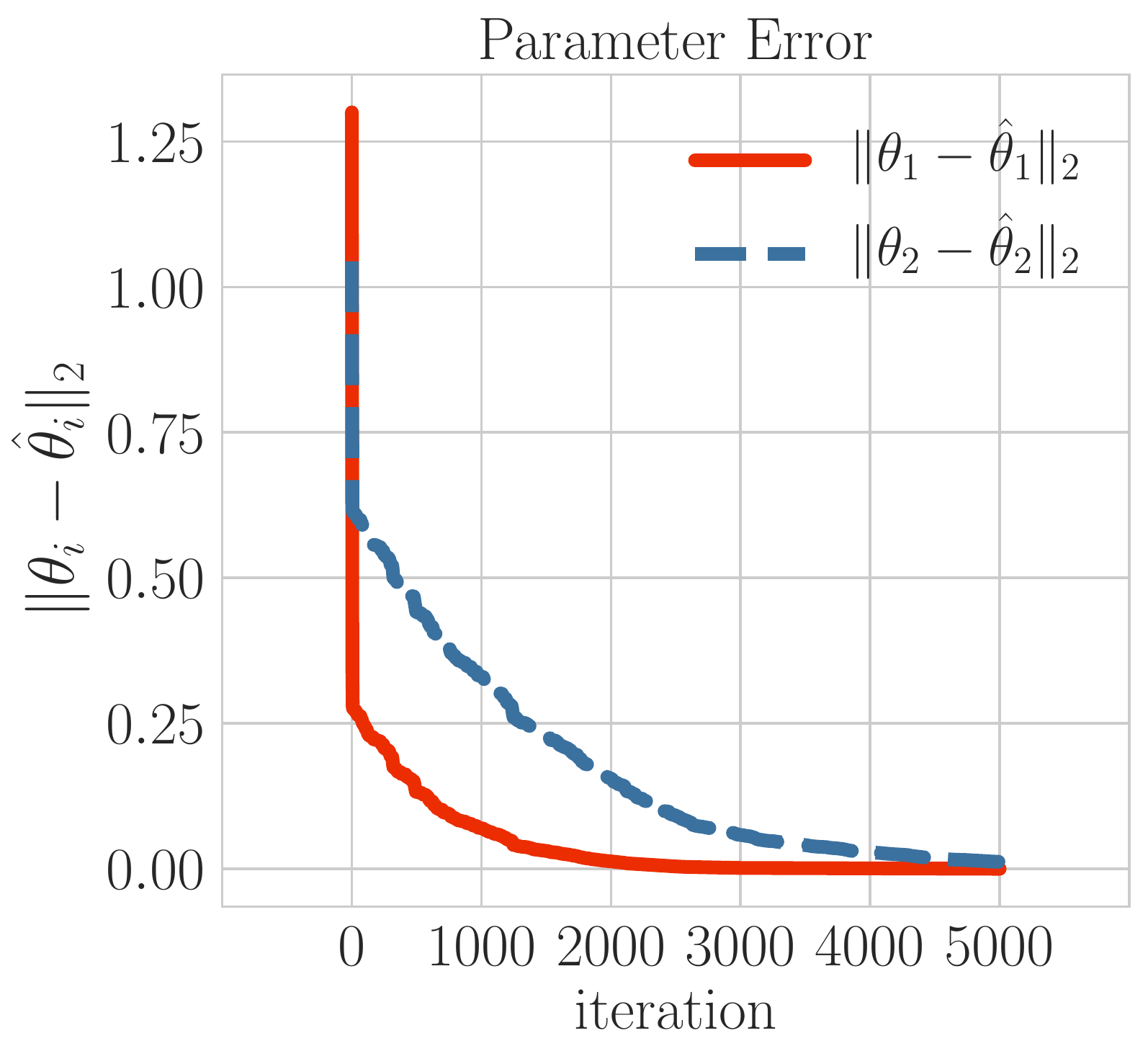}}\hfill$\
    $
    \end{center}
    \caption{\footnotesize (a) Agent response and (b) parameter estimation
        error.
        While the
agents' responses converge relatively quickly to the desired responses (within a
range dictated by the statistics of the noise term $\tau$), the parameter estimation takes quite a
bit longer to converge within $1$e-$2$. This is in part due to the fact that the response
converges quickly and the \emph{information} that can be attained from the
responses is limited; hence, it takes a while to improve the estimate. }
\end{figure}

\subsubsection{True Model Known to Planner}
 We consider each firm $i$ to have a linear marginal revenue function given by
\begin{equation}\label{eq:linmargrev}
    M_i(x^k,\tau^k)= \theta_{i,1}^\ast x_1^k+\theta_{i,2}^\ast
    x_2^k+\theta_{i,i}^\ast x_i^k+ \tau^k,
\end{equation}
and the firm price to evolve according to a gradient play update
\begin{equation}\label{eq:grad_play}
x_i^{k+1} = x_i^k + \zeta_i^k (M_i(x_i^k, x_{-i}^k, \tau^k) + \langle \Psi_i(x_i^k, x_{-i}^k), \alpha_i^{k} \rangle)
\end{equation}
where $\zeta_i^k>0$ is the learning rate which we choose to ensure stability of the
process. We let the incentive basis functions be the following set of Gaussian
radial basis functions:
\begin{equation}\label{eq:inc_basis}
    \left\{\begin{array}{lcl}
\Psi_{i,1}(x_i^k, x_{-i}^k) &=& \exp(-\kappa (x_i^k - x_i^d)^2) \\
\Psi_{i,2}(x_i^k, x_{-i}^k) &=& \exp(-\kappa (x_i^k + x_i^d)^2) \\
\Psi_{i,3}(x_i^k, x_{-i}^k) &=& \exp(-\kappa (x_i^k)^2)\end{array}\right.
\end{equation}
with $\kappa = 0.01$.
In order to test the theoretical results, in this example we assume the
regulator has knowledge of the true nominal basis functions and assumes a
gradient play model for the firms,
    $\Phi(x_1^k,x_2^k)=\{x_1^k,x_2^k\}$,
and updates its estimate of the firm prices according to 
\begin{equation*}
    x_i^{k+1} = x_i^k + \zeta_i^k (\braket{\Phi_i(x_i^k,
        x_{-i}^k)}{\hat{\theta}_i^k} + \tau^k + \langle \Psi_i(x_i^k, x_{-i}^k), \alpha_i^{k} \rangle).
\end{equation*}
For this simulation we use $\tau\sim\mc{N}(5,0.25)$, 
$\theta_1^\ast=(-1.2,-0.5)$, and $\theta_2^\ast=(0.3,-1)$, and for all Bertrand examples we set the desired prices at $(x_1^d, x_2^d) = (5, 7)$. 

In Fig.~\ref{fig:true_gp_responses}, we show the firm prices and the regulator's
estimates of the firm prices over the first $100$ iterations of the simulation.
The firm prices $x_i$ almost immediately converge to the desired
prices $x_i^d$. Likewise, the regulator's estimates of the firm prices
$\hat{x}_i$ rapidly converge to the true firm prices. We remark that in our
simulations we observe that even with fewer incentive basis functions, the firm
prices can still be pushed to the desired prices. Correspondingly, we find that
increasing the number of incentive basis functions can increase the speed of
convergence to the desired prices as well as mitigate the impact of noise due to
more flexibility in the choice of incentives.

In Fig.~\ref{fig:true_gp_error}, we show that the regulator's estimates of the
parameters of the marginal revenue function for each firm converge to the true
parameters. While in this example the speed of estimation convergence is
considerably slower than the convergence of the firm prices to the desired
prices, we observe multiple factors that can significantly increase the speed of
estimation convergence. Namely, as the variance of the economic indicator term
$\tau$ increases, consequently providing more information at each iteration, the
convergence speed also increases. Moreover, the stability of the process,
dictated by the parameters of the marginal revenue functions, is positively
correlated with the speed of convergence. We also note that the problem is
persistently exciting and the expected loss decreases rapidly to zero.
\begin{figure}[t]
    \centering
    \includegraphics[width=0.28\textwidth]{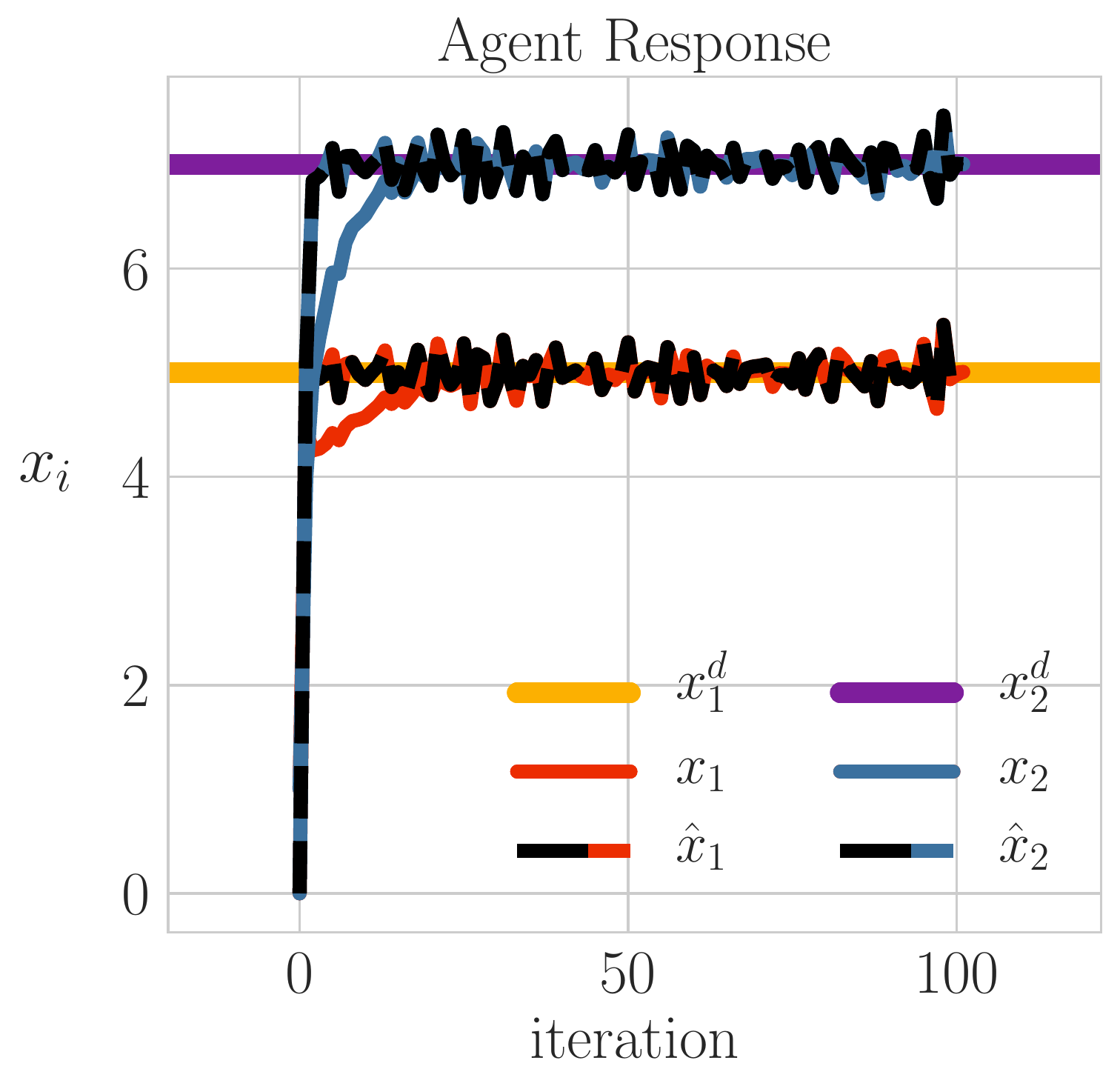}
    \caption{Prices for best response (gradient and
    fictitious play are similar).}
  \label{fig:br_return}
  \end{figure}
 \begin{figure}[t]
    \centering
 \hfill\subfloat[][]{\label{fig:lin_estimation}\includegraphics[width=0.2875\textwidth]{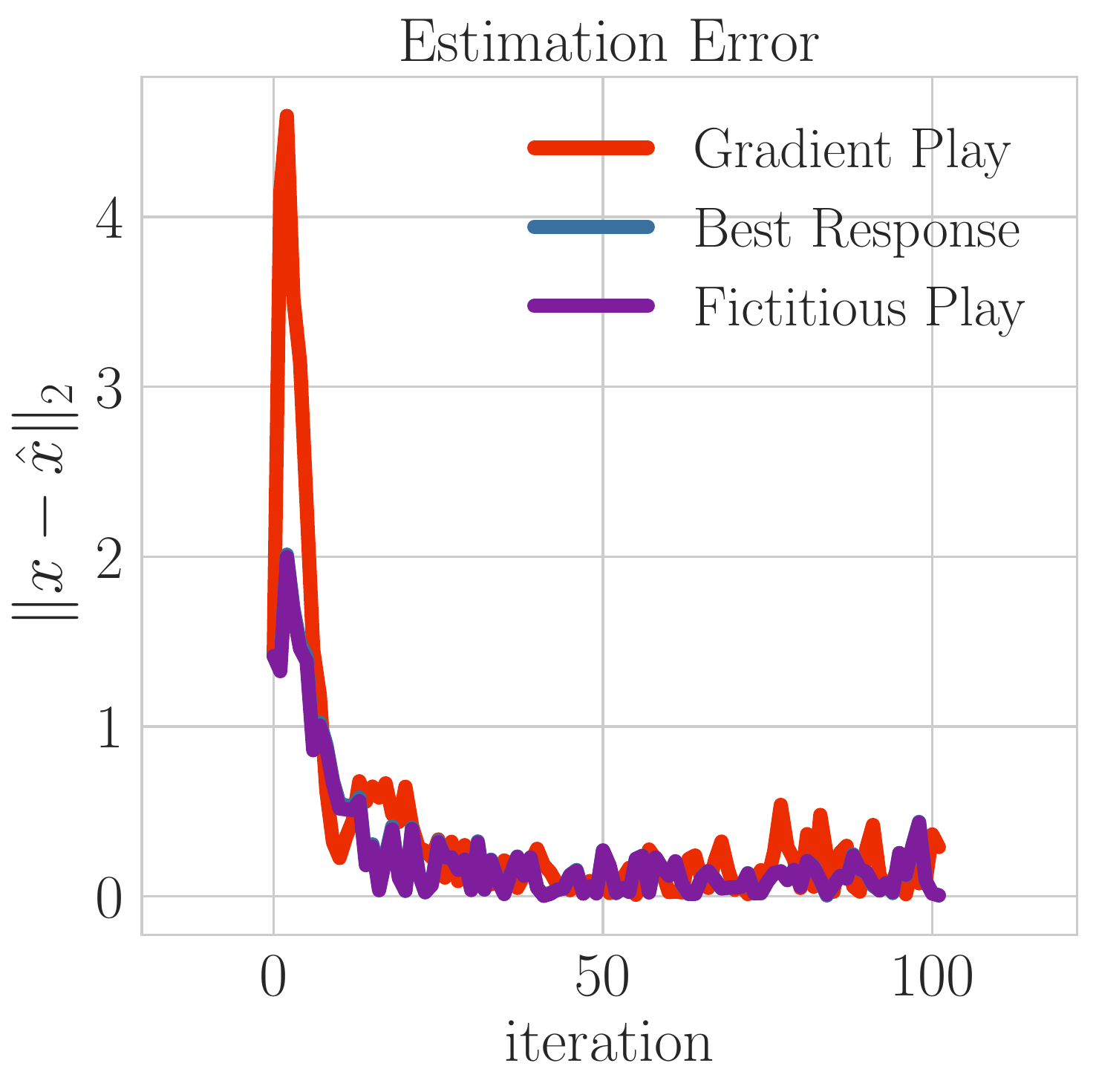}}\hfill
    \subfloat[][]{\label{fig:lin_incentive}\includegraphics[width=0.29\textwidth]{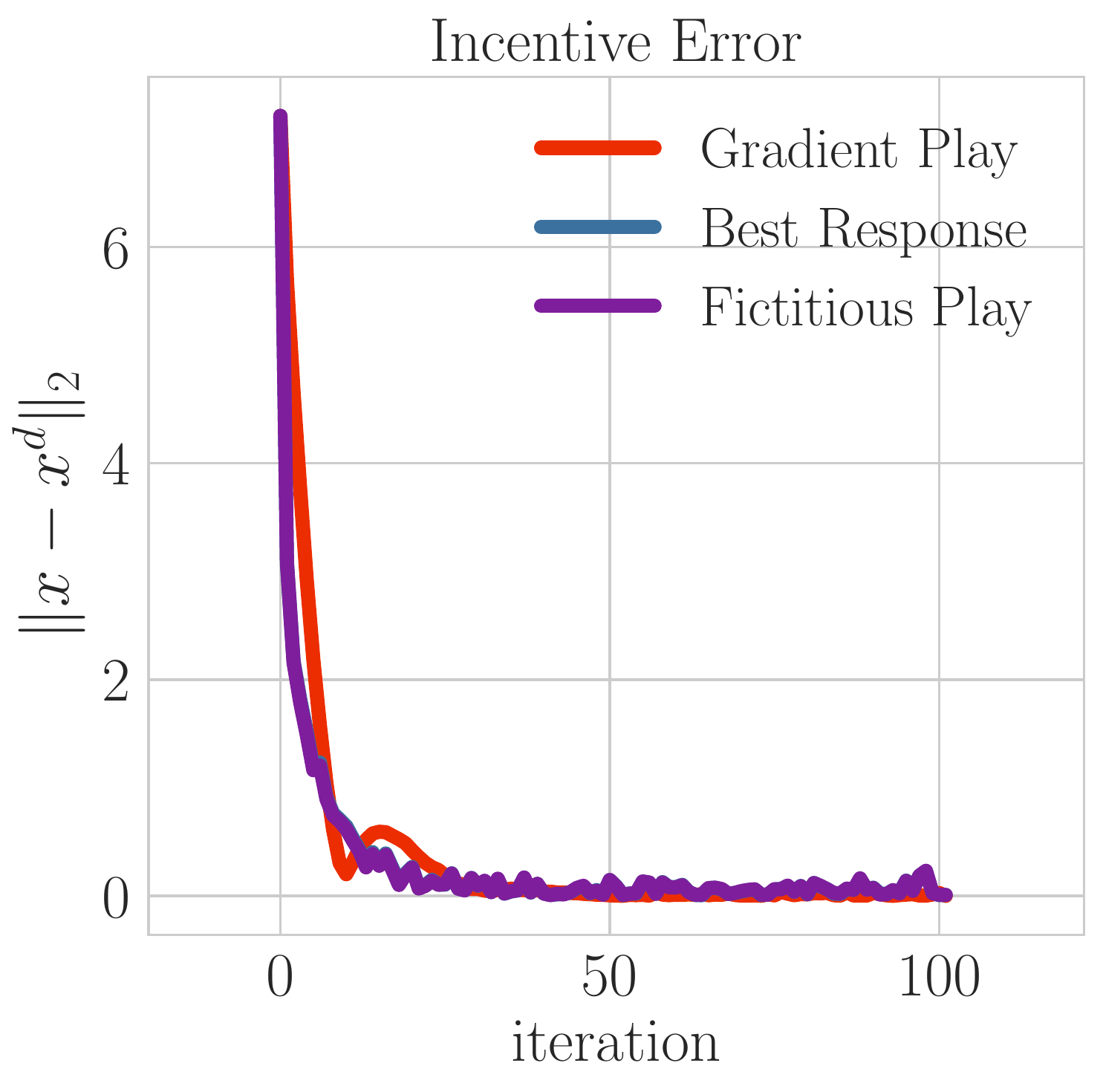}}\hfill$\
    $
    \caption{(a) Response estimation error for different update methods and (b) difference
between desired and true response for different update methods. Note that in (a)
and (b) the best response and fictitious play methods are nearly identical; this
is due to the fact that fictitious play is best response to the historical
average. We remark that it takes longer for fictitious play to converge given
the incentive and noise at each iteration.}
  \label{fig:br_return}
  \end{figure}

\subsubsection{Planner Agnostic to Myopic Update Rule}
We now transition to examples which demonstrate that our framework is suitable
for problem instances that go beyond our theoretical results. In particular, we
find even if the planner does not have access to the true nominal basis
functions that govern the dynamics of agent play, agent responses can still be
driven to a desired response for several forms of myopic updates. In these
examples, we again consider firms to have the linear marginal revenue function
from \eqref{eq:linmargrev}. We then allow the firm price to evolve according to
either the gradient play update from \eqref{eq:grad_play}, a best response update given by
\begin{equation*}
    x_i^{k+1} = \arg\max_{x_i} \ R_i(x_i,x^k_{-i}, \tau^k) + \gamma_i^k(x_i,
    {x}_{-i}^{k}),
\end{equation*}
or a fictitious play update given by
\begin{align*}
    x_i^{k+1}& =
    \arg\max_{x_i} \ R_i(x_i, \bar{x}_{-i}^{(l,k)}, \tau^k)+ \gamma_i^k(x_i,
    \bar{x}_{-i}^{(l,k)})
\end{align*}
where $\bar{x}_{-i}^{(l,k)}$ is the average over the history from $l$ to $k$
with $0\leq l\leq k$. For all simulations we choose to use the complete history,
i.e.~$l = 0$.

We choose the set of incentive basis functions to be same as in
\eqref{eq:inc_basis}, but now consider the nominal basis functions to not
include the true basis functions and instead be the following set of Gaussian
radial basis functions:
\begin{equation}\label{eq:nominal}
    \left\{\begin{array}{lcl}\Phi_{i,1}(x_i^k, x_{-i}^k) &=& \exp(-\kappa (x_i^k)^2) \\
\Phi_{i,2}(x_i^k, x_{-i}^k) &=& \exp(-\kappa (x_i^k - x_{-i}^k)^2) \\
\Phi_{i,3}(x_i^k, x_{-i}^k) &=& \exp(-\kappa (x_{-i}^k)^2) \end{array}\right.
\end{equation}
with $\kappa = 0.01$ as before. In this example, the regulator uses an
 update to
its estimate of the prices, given by 
\begin{equation}\label{eq:est_cost}
x_i^{k+1} = \langle \Phi_i(x_i^k, x_{-i}^k), \hat{\theta}_i^k \rangle + \tau^k +
\langle \Psi_i(x_i^k, x_{-i}^k), \alpha_i^{k} \rangle, 
\end{equation}
which is agnostic to the firm update method.
 We use $\tau\sim\mc{N}(5,0.1)$ and, as before, let the parameters of the marginal revenue functions be $\theta_1^\ast=(-1.2,-0.5)$, and $\theta_2^\ast=(0.3,-1)$. 

In Fig.~\ref{fig:br_return}, we show the firm prices and the regulator's
estimates of the firm prices when the firms use best response (each of the other methods have similar plots).
As was the case when the regulator had knowledge of the true nominal basis
functions, the regulator is able to quickly
estimate the true firm prices accurately and push the firm prices to the desired
prices. 
Figs.~\ref{fig:lin_estimation}
and \ref{fig:lin_incentive} show that the
regulator's estimate of the firm prices converges to the true firm
prices---i.e.~$\|\hat{x}^k- x\|_2\rar 0$---and the true firm responses converge
to the desired responses---i.e.~$\|x^k-x^d\|_2\rar 0$---at nearly the same rates. 



In Fig.~\ref{fig:br_cost_comps}, we show the nominal and incentive cost components that form the
regulator's estimate of the firm prices when the firms are use best response. Analogous to the firm prices, we note that
these components are similar for each of the firms' update methods. We highlight this to point out that the nominal cost and incentive cost components do not simply cancel one another out, indicating that a portion of the true dynamics are maintained. 

\subsubsection{Nonlinear Marginal Revenue}
To conclude our examples, we explore a Bertrand competition in which each firm has a nonlinear marginal revenue function given by
\begin{equation*}\label{eq:nonlinmargrev}
    M_i(x^k,\tau^k)=\log(x_i^k)+ \theta_{i,1}^\ast x_1^k+\theta_{i,2}^\ast
    x_2^k+\theta_{i,i}^\ast x_i^k+\theta_{i,3}+ \tau^k+1.
\end{equation*}
Here, we let the firm price evolve according to the gradient play update from
\eqref{eq:grad_play} with the nonlinear marginal revenue function and let the
regulator use the agnostic method to update its estimate of the firms' prices
from \eqref{eq:est_cost} with the nominal basis functions from
\eqref{eq:nominal} and the incentive basis functions from \eqref{eq:inc_basis}.
We use the same noise parameters as in the preceding examples and
and let the parameters of the marginal revenue functions be $\theta_1^\ast=(-1.2,-0.5,7.5)$, and $\theta_2^\ast=(0.3,-1, 1.5)$. 
\begin{figure*}
    \begin{center}
        \subfloat[][]{\label{fig:br_cost_comps}\includegraphics[width=0.37\textwidth]{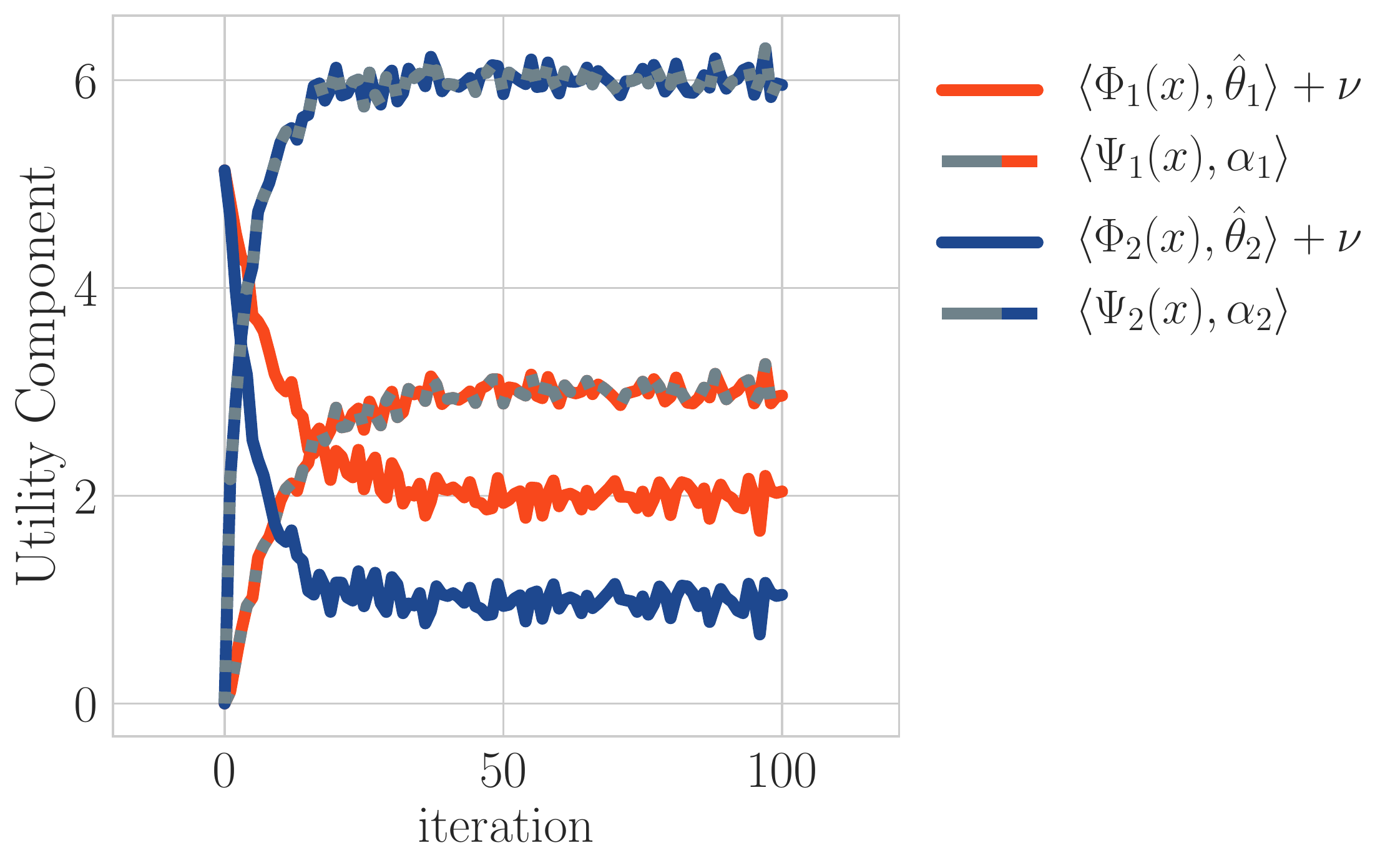}}
        \hfill
       \subfloat[][]{\label{fig:nonlin_est}\includegraphics[width=0.37\textwidth]{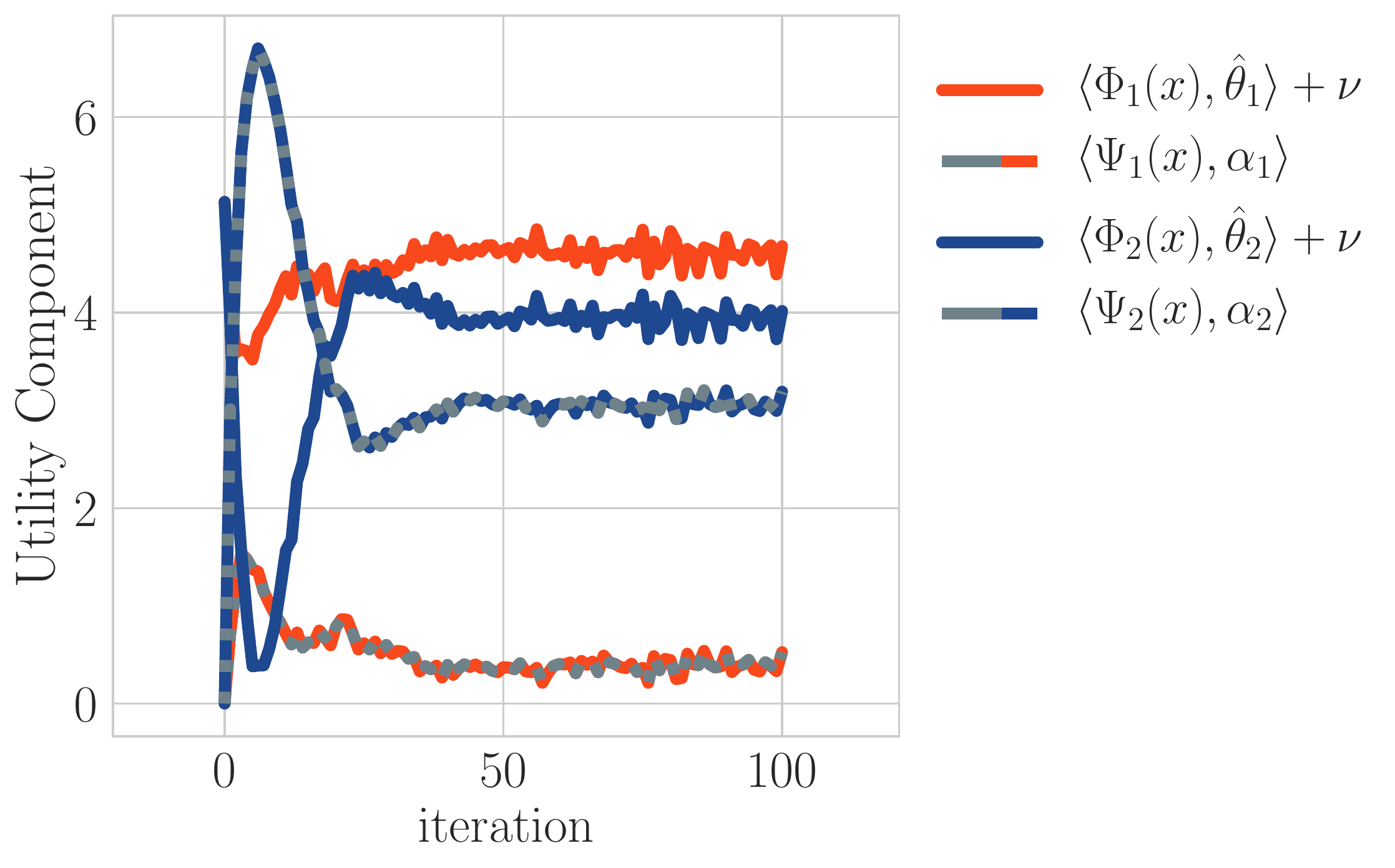}}\hfill
       \subfloat[][]{\label{fig:nonlin_returns}\includegraphics[width=0.25\textwidth]{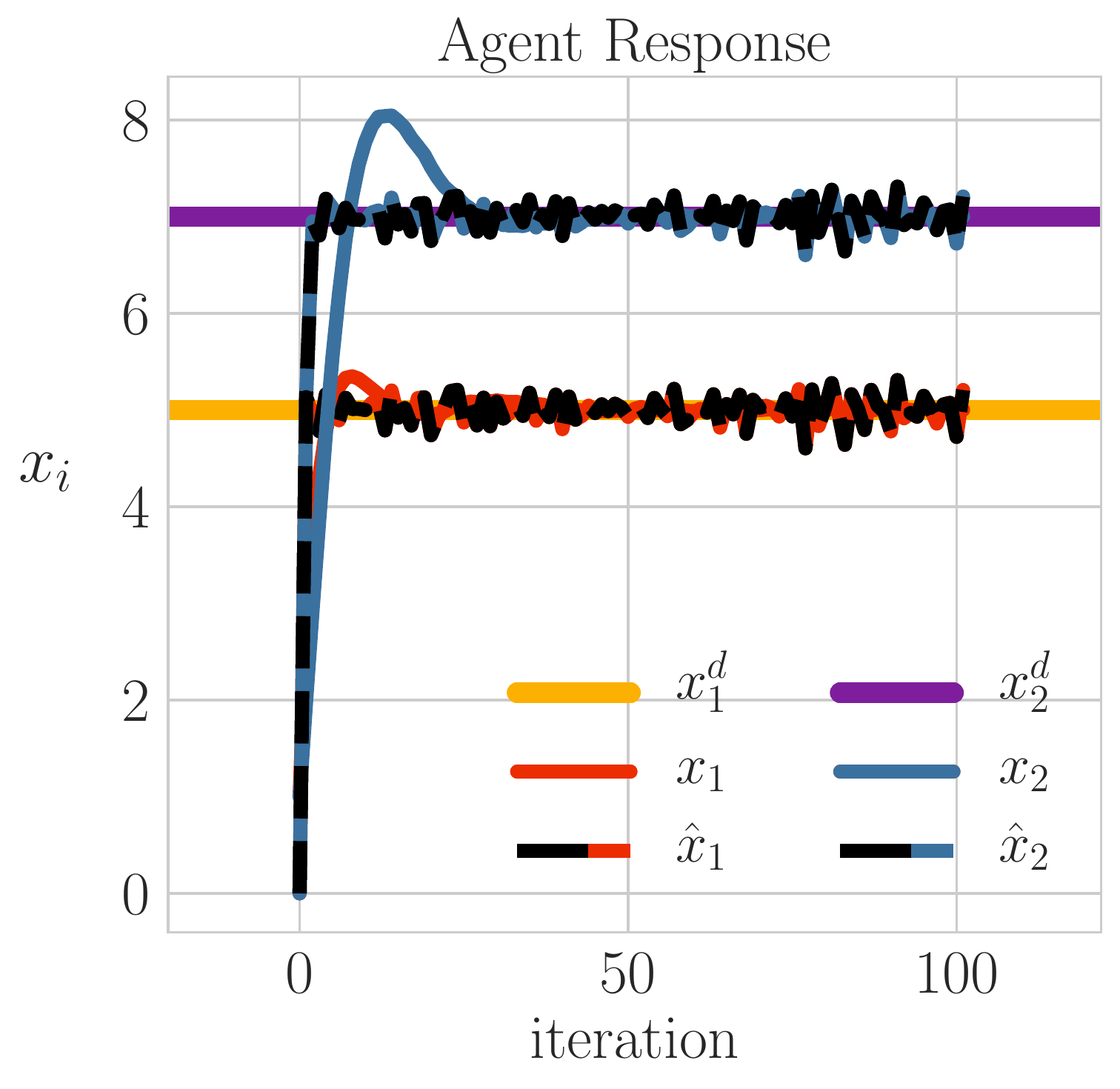}}
    \end{center}
    \caption{\footnotesize  Estimated cost components for each firm when
        the firms use the best response update method: (a) linear marginal
        revenue and (b) non-linear marginal revenue.
        (c) Agent responses for gradient play with a
    non-linear marginal revenue functions.}
\end{figure*}

In Fig.~\ref{fig:nonlin_est}, we show the estimated revenue and incentive components for the
two firms.
In Fig.~\ref{fig:nonlin_returns}, we observe that the response of the firms to the regulator is marginally different from what we observed in previous examples when the firms had linear marginal revenue functions, but nonetheless the regulator's estimates of the firm prices converge to the true firm prices and the true firm prices converge to the desired prices.


\section{Conclusion}
\label{sec:discussion}
We present a new method for adaptive incentive design when a planner faces
competing agents of unknown type. Specifically, we provide an algorithm for
learning the agents' decision-making process and updating incentives. We provide
convergence guarantees on the algorithm. We show that under reasonable
assumptions, the agents' true response is driven to the desired response and,
under slightly more restrictive assumptions, the true preferences can be learned
asymptotically. We provide several numerical examples that both verify the
theory as well as demonstrate the performance when we relax the theoretical assumptions.

\appendix
\section{Proof of Lemma~\ref{lem:tvpmaps}}
\label{app:proofs}
The following proof uses Young's inequality: 
\eq{eq:young}{\textstyle
v_1^Tv_2\leq \|v_1\|_\ast\|v_2\|\leq \frac{1}{2}\left(
\frac{\|v_1\|_\ast^2}{\nu}+\nu \|v_2\|^2 \right)}
for any $v_1,v_2\in \mb{R}^m$.

\begin{proof}[Proof of Lemma~\ref{lem:tvpmaps}]
  The proof is the essentially the same as the proof for \cite[Lemma 2.1]{nemirovski:2009aa} with a few modifications. 
  
  Let $\tak{k}\in (\Theta^{k})^\circ$ and $\bth=P_{\tak{k}}^{k+1}(g)$. Note that
  \eq{eq:pup}{
  \bth\in \argmin_{\theta'\in
      \Theta^{k+1}}\left\{\braket{g}{\theta'-\tak{k}}+V(\tak{k},\theta')\right\}}
or equivalently, 
 \eq{eq:pup2}{
  \bth\in \argmin_{\theta'\in
      \Theta^{k+1}}\left\{\beta(\theta')-\braket{\nabla\beta(\tak{k})-g}{\theta'}\right\}}
  where the latter form tells us that $\beta$ is differentiable at $\bth$ and
  $\bth\in(\Theta^{k+1})^\circ$. Since $\nabla_2V(\tak{k},
  \bth)=\nabla\beta(\bth)-\nabla\beta(\tak{k})$, the optimality conditions
  for~\eqref{eq:pup2} imply that 
  \begin{equation}      \braket{\nabla\beta(\bth)-\nabla\beta(\tak{k})+g}{\bth-\theta}\leq 0, \ \ \forall
  \theta\in \Theta^{k+1}
  \label{eq:optcon}
  \end{equation}
  Note that this is where the proof of \cite[Lemma 2.1]{nemirovski:2009aa}
  and the current proof are different. The above inequality holds here for all
  $\theta\in \Theta^{k+1}$ whereas in the proof of
  \cite[Lemma 2.1]{nemirovski:2009aa}---using the notation of the current Lemma---the
  inequality would have held for all $\theta\in \Theta^{k}$. In particular, we
  need the inequality to hold for $\theta^\ast$ and it does since by assumption
   $\theta^\ast\in \Theta^{k+1}$ for each
  $k$.

  First,
$\braket{\nabla \beta(\theta^k)}{\theta^{k+1}-\theta^k}\leq
  \beta(\theta^{k+1})-\beta(\theta^k)$.
  Hence, for $\theta^\ast\in\Theta^{k+1}$, we have that
  \begin{align*}
    &V(\bth, \ts)-V(\tak{k}, \ts)\\
    &=\beta(\ts)-\braket{\nabla
    \beta(\bth)}{\ts-\bth)}-\beta(\bth)\notag\\
    &\qquad-(\beta(\ts)-\braket{\nabla\beta(\tak{k})}{\ts-\tak{k}}-\beta(\tak{k}))\\
    &=
    \braket{\nabla\beta(\bth)-\nabla\beta(\tak{k})+g}{\bth-\ts}+\braket{g}{\ts-\bth)}\\
    &\leq \braket{g}{\ts-\bth}-V(\tak{k}, \bth)
    \label{eq:lem-ineqs}
  \end{align*}
  where the last inequality holds due to~\eqref{eq:optcon}. By~\eqref{eq:young},
  we have that 
  \eq{eq:yres}{
      \textstyle  \braket{g}{\tak{k}-\bth}\leq
  \frac{1}{2\nu}\|g\|_\ast^2+\frac{\nu}{2}\|\tak{k}-\bth\|^2.}
  Further, $\frac{\nu}{2}\|\tak{k}-\bth\|^2\leq V(\tak{k}, \bth)$ since
  $V(\tak{k}, \cdot)$ is strongly convex. Thus, 
  \begin{align*}
    &V(\bth, \ts)-V(\tak{k}, \ts)\\
    &\qquad\leq \braket{g}{\ts-\bth}-V(\tak{k}, \bth)\\
   &\qquad=\braket{g}{\ts-\tak{k}}+\braket{g}{\tak{k}-\bth}-V(\tak{k}, \bth)\\
   &\qquad\leq\textstyle \braket{g}{\ts-\tak{k}}+\frac{1}{2\nu}\|g\|_\ast^2
  \end{align*}
  so that 
  \[\textstyle V(P_{\theta^{(k)}}^{k+1}(g), \ts)\leq V(\theta^{(k)},
  \ts)+\braket{g}{\ts-\theta^{(k)}}+\frac{1}{2\nu}\|g\|_\ast^2.\]
\end{proof}

\section{Proof of Theorem~\ref{thm:noise-p}}
\begin{proposition}[\cite{neveu:1975aa}] Let $\{x_t\}$ be a zero conditional
  mean sequence of random variables adapted to $\{\F_t\}$. If $\sum_{t=0}^\infty
  \frac{1}{t^2}\mb{E}[x_t^2|\F_{t-1}]<\infty$ a.s, then
  $\lim_{N\rar\infty}\frac{1}{N}\sum_{t=1}^Nx_t=0$ a.s.
  \label{prop:neveu}
\end{proposition}

\begin{proof}[Proof of Theorem~\ref{thm:noise-p}]
 Starting from Lemma~\ref{lem:tvpmaps}, we have
 \begin{align*}
   \mb{E}[V_{t+1}(\ts_i)|\F_t]&\leq
   V_t(\ts_i)-\eta_t\braket{\xi_i^t}{\Delta\tit{i}{t}}(\mb{E}[y_i^{t+1}-\braket{\xi_i^t}{\tit{i}{t}}
   \textstyle|\F_t])\\
   &\qquad+\frac{\eta_t^2}{2\nu}\|\xi_i^t\|_\ast^2\left((\|\mb{E}[y_i^{t+1}-(\xi_i^t)^T\tit{i}{t}|\F_t]\|
    )^2+\sigma^2\right)\notag\\
    &\leq\textstyle V_t(\ts_i)-\big(
    \eta_t-\frac{\eta_t^2\|\xi_i^t\|_\ast^2}{2\nu} \big)\left(
    \|\mb{E}[y_i^{t+1}-(\xi_i^t)^T\tit{i}{t}|\F_t]\|
    \right)^2\textstyle+\frac{\eta_t^2}{2\nu}\|\xi_i^t\|_\ast^2\sigma^2\\
    &\leq  \textstyle V_t(\ts_i)-\big(
    \eta_t-\frac{\eta_t^2}{2\nu}\tilde{c}_1 \big)\big(
    \|\mb{E}[y_i^{t+1}-(\xi_i^t)^T\tit{i}{t}|\F_t]\|
    \big)^2\textstyle+\frac{\eta_t^2}{2\nu}\tilde{c}_1\sigma^2.
   \label{eq:noiseeq-1}
 \end{align*}
 By the assumptions that $\eta_t-\frac{\eta_t^2}{2\nu}\tilde{c}_1>0$ and
 $\sum_{t=1}^\infty\eta_t^2<\infty$, we can use the fact that
 $y_i^{t+1}=\braket{\xi_i^t}{\theta_i^\ast}+w_i^{t+1}$ and apply the almost supermartingale
 convergence theorem~\cite{robbins:1985aa} to get that  
 \begin{equation*}\textstyle
 \sum_{t=1}^\infty \big(
    \eta_t-\frac{\eta_t^2}{2\nu}\tilde{c}_1 \big)\left(
    \|\mb{E}[w_i^{t+1}-\braket{\xi_i^t}{\tit{i}{t}-\theta_i^\ast}|\F_t]\|
    \right)^2<\infty
    \label{eq:sumit}
  \end{equation*}
 a.s.
    and that $V_t(\ts_i)$
    converges a.s.

    Now, we argue~\eqref{eq:errorcon-2} holds; the argument follows that which is
    presented in~\cite[Chapter 8]{goodwin:1984aa}. To do this, we first show that 
    \begin{equation}\textstyle
%
      \lim_{T\rar\infty} \frac{1}{T}\sum_{t=0}^{T-1}
      \|\braket{\xi_i^{t}}{\tit{i}{t}-\theta_i^\ast}\|^2=0 \ \ \text{a.s.}
    \label{eq:8576}
  \end{equation}
   Note that
    \eqref{eq:sumit} implies that
    \begin{equation}\textstyle
    \lim_{T\rar\infty}\sum_{t=0}^{T-1}
    \frac{1}{r_t}\| \braket{\xi_i^t}{\tit{i}{t}-\theta_i^\ast}   \|^2<\infty\ \ \text{a.s.}
      \label{eq:sumit-v2}
    \end{equation}
  Where $r_t=(\eta_t-\frac{\eta_t^2}{2\nu}\tilde{c}_1)^{-1}$. 
    Suppose that $r_t$ is bounded---i.e.~there exists $K_3$ such that
    $r_t<K_3<\infty$. In this case, it is
    immediate from \eqref{eq:sumit-v2} that
\begin{equation}\textstyle
\lim_{T\rar\infty}\frac{1}{K_3}\sum_{t=0}^{T-1}
{\|\braket{\xi_i^t}{\tit{i}{t}-\theta_i^\ast}   \|^2}<\infty\ \ \text{a.s.}
      \label{eq:sumit-v3}
    \end{equation}
    so that \eqref{eq:8576} follows trivially.
On the other hand, suppose $r_t$ is unbounded. Then we can apply Kronecker's
Lemma~\cite{kumar:1986aa} to conclude that
\eq{eq:kronekit}{\textstyle
    \lim_{T\rar\infty}\frac{1}{r_T}\sum_{t=0}^{T-1}{\|\braket{\xi_i^{t}}{\tit{i}{t}-\theta_i^\ast}
\|^2}=0\ \ \text{a.s.}}
Hence, from \eqref{eq:rbound-22}, we have 
\begin{equation}\textstyle
    \lim_{T\rar\infty}\frac{\frac{1}{T}\sum_{t=0}^{T-1}\|\braket{\xi_i^t}{\tit{i}{t}-\theta_i^\ast}
\|^2}{K_1+\frac{K_2}{T}\sum_{t=0}^{T-1}\|\braket{\xi_i^t}{\tit{i}{t}-\theta_i^\ast}   \|^2}=0\ \ \text{a.s.}
%
  \label{eq:rbound-2}
\end{equation}
so that \eqref{eq:8576} follows immediately. Note that
\begin{align*}
  \mb{E}[(y_i^{t+1}-(\xi_i^t)^T\tit{i}{t})^2|\F_t]\textstyle&=\mb{E}[(y_i^{t+1}-\braket{\xi_i^{t}}{\tit{i}{t}}-w_i^{t+1}+w_i^{t+1})^2|\F_t]\\
 &=\mb{E}[(y_i^{t+1}-\bbraket{\xi_i^{t}}{\theta_i^{(t)}}-w_i^{t+1})^2+(w_i^{t+1})^2\\
 &\qquad+2\bbraket{y_i^{t+1}-\bbraket{\xi_i^{t}}{\theta_i^{(t)}}-w_i^{t+1}}{w_i^{t+1}}|\F_t]\notag
\end{align*}
Since $y_i^{t+1}-w_i^{t+1}$ and $\braket{\xi_i^t}{\tit{i}{t}}$ are $\F_t$--measurable 
 and $\mb{E}[w_i^{t+1}|\F_t]=0$ a.s., we have 
\begin{align*}
    \mb{E}[(y_i^{t+1}-\braket{\xi_i^t}{\tit{i}{t}})^2|\F_t]&=(y_i^{t+1}-\braket{\xi_i^{t}}{\tit{i}{t}}-w_i^{t+1})^2-\mb{E}[(w_i^{t+1})^2|\F_t].
  \label{eq:expect-w}
\end{align*}
Replacing $y_i^{t+1}=\braket{\xi_i^t}{\theta_i^\ast}+w_i^{t+1}$ and using
\eqref{eq:8576}, we see that \eqref{eq:errorcon-2} holds since
$\mb{E}[(w_i^{t+1})^2|\F_t]=\sigma^2$
a.s.

Finally, if   $\sup_t \mb{E}[(w_i^{t+1})^4| \F_t]<+\infty$ almost surely, then
by
Proposition~\ref{prop:neveu}, we have
\begin{equation*}\textstyle
  \lim_{T\rar\infty}
  \frac{1}{T}\sum_{t=0}^{T-1}\|w_i^{t+1}-\braket{\xi_i^{t}}{\tit{i}{t}-\theta_i^\ast}\|^2=\sigma^2\ \ \text{a.s.}
\label{eq:errorcon-1}
\end{equation*}
which concludes the proof.
\end{proof}


%


\bibliographystyle{IEEEtran}
\bibliography{2017TACLTrefs}

\end{document}